\documentclass[12pt]{article}

\newcommand{\blind}{1}

\usepackage{subfiles}

\usepackage{amsmath}

\usepackage{times}
\usepackage{bm}
\usepackage{natbib}
\usepackage{subcaption}
\usepackage{float}
\usepackage{overpic}
\usepackage{amssymb}
\usepackage[ruled,vlined]{algorithm2e}
\usepackage{xcolor}
\usepackage{multirow}
\allowdisplaybreaks

\DeclareMathOperator{\No}{N}

\protected\def\[#1\]{\begin{equation}\begin{aligned}#1\end{aligned}\end{equation}}
\protected\def\(#1\){\begin{equation*}\begin{aligned}#1\end{aligned}\end{equation*}}

\usepackage{float}

\usepackage{amsthm}
\usepackage{amsmath}
\usepackage{natbib}
\usepackage[colorlinks,citecolor=blue,urlcolor=blue,filecolor=blue,backref=page]{hyperref}
\usepackage{graphicx}

\numberwithin{equation}{section}
\theoremstyle{plain}
\newtheorem{theorem}{Theorem}[section]

\newtheorem{remark}{Remark}[section]

\begin{document}

\def\spacingset#1{\renewcommand{\baselinestretch}%
{#1}\small\normalsize} \spacingset{1}

\if1\blind
{
        \title{Gibbs Sampling using Anti-correlation Gaussian Data Augmentation, with Applications to L1-ball-type  Models}
\author{
Yu Zheng \footnote{Department of Statistics, University of Florida, {zheng.yu@ufl.edu}}
\qquad
Leo L. Duan\footnote{Department of Statistics, University of Florida, {li.duan@ufl.edu}}
 }
 \date{}
  \maketitle
} \fi

\bigskip
\begin{abstract}
L1-ball-type priors are a recent generalization of the spike-and-slab priors. By transforming a continuous precursor distribution to the L1-ball boundary, it induces exact zeros with positive prior and posterior probabilities. With great flexibility in choosing the precursor and threshold distributions, we can easily specify models under structured sparsity, such as those with dependent probability for zeros and smoothness among the non-zeros.
Motivated to significantly accelerate the posterior computation, we propose a new data augmentation that leads to a fast block Gibbs sampling algorithm. The latent variable, named ``anti-correlation Gaussian'', cancels out the quadratic exponent term in the latent Gaussian distribution, making the parameters of interest conditionally independent so that they can be updated in a block. Compared to existing algorithms such as the No-U-Turn sampler, the new blocked Gibbs sampler has a very low computing cost per iteration and shows rapid mixing of Markov chains. We establish the geometric ergodicity guarantee of the algorithm in linear models. Further, we show useful extensions of our algorithm for posterior estimation of general latent Gaussian models, such as those involving multivariate truncated Gaussian or latent Gaussian process.
\end{abstract}

\noindent%
\vfill

\newpage
\spacingset{1} %

\maketitle

\section{Introduction}
There is a large literature on Bayesian sparse models and the inspired algorithms for estimating the posterior.
Originally motivated for solving the variable selection problem in linear regression, spike-and-slab priors (or, ``discrete spike-and-slab'', \citealt{tadesse2021handbook}) have the marginal form of a two-component mixture for each parameter element: one component (spike) from a point mass at zero, and the other (slab) from a continuous distribution \citep{mitchell1988bayesian}; the continuous elements can be independent or dependent a priori.
For linear regression with Gaussian errors, very efficient Markov chain Monte Carlo (MCMC) algorithms have been developed. When the slab prior distribution follows a Gaussian, the Stochastic Search Variable Selection (SSVS) algorithm \citep{george1995stochastic} exploits the posterior conjugacy and samples from the marginal posterior of the binary inclusion variables. SSVS is a Gibbs sampler that tries to flip each binary inclusion variable one at a time from a Bernoulli full conditional distribution; at the end of each iteration, it samples the regression coefficients given the vector of binary inclusion variables. As an alternative to using the marginal posterior, the Orthogonal Data Augmentation (ODA) algorithm \citep{ghosh2011rao} introduces an augmented design matrix (along with augmented responses) to append the observed design matrix, such that the Gram matrix becomes diagonal hence easily invertible, enabling the use of two-block Gibbs sampler. The ODA algorithm can be extended to some generalized linear models if the latent Gaussian has a fixed variance, such as probit regression \citep{albert1993bayesian}. Further, for design matrix that is high-dimensional or contains highly correlated predictors, various algorithms have been proposed such as shotgun algorithm \citep{hans2007shotgun}, parallel tempering \citep{bottolo2010evolutionary}, correlation-based search \citep{kwon2011efficient}, two-parameter flipping Metropolis-Hastings algorithm under $g$-prior slab \citep{yang2016computational}. In the meantime, there is a comparably vast literature on continuous shrinkage priors with excellent performance for the task of variable selection \citep{george1995stochastic,rovckova2018spike,polson2010shrink,carvalho2010horseshoe,piironen2017sparsity,armagan2013generalized,bhattacharya2015dirichlet,bai2018beta}. Since our focus is on the posterior with exact zeros, for brevity, we will skip the detail of continuous shrinkage.

With the rich literature, there is a recent interest in structured sparsity \citep{hoff2017lasso,griffin2023structured} that has inspired new extensions of sparsity priors. Specifically, the sparsity is ``structured'' in the sense that: (i) the occurrences of zeros could be dependent, according to some temporal, spatial, or group structure; (ii) the non-zeros could have some correlation structure, such as smoothness over some spatial domain. We now provide a few examples. In the task of change-point detection, one may model a time series as having the mean increments to be sparse over a continuous period of time so that the mean function would become a step function \citep{tibshirani2005sparsity,betancourt2017bayesian}. In the scalar-on-image regression, one may model the regression coefficients to be spatially smooth for those non-zeros, while being zero over several continuous regions \citep{kang2018scalar}. For these models, a critical computational challenge arises that the above existing algorithms for spike-and-slab priors cannot be applied directly here, due to the lack of conjugacy or fixed variance for latent Gaussian. Because of the correlations among the elements of the parameter, updating the full conditional of one element at a time suffers from slow mixing of Markov chains.
 To facilitate the model development and posterior estimation under structured sparsity, \citep{xu2020bayesian} propose the L1-ball priors: starting with a continuous random variable $\beta\in\mathbb{R}^p$ from a ``precursor'' distribution (that could be correlated or smooth over some input space), one projects $\beta$ to the boundary of the L1 ball. This projection produces $\theta$ containing exact zeros with positive probability (in both prior and posterior). To explain the gist of the idea, we focus on the equivalent generative process using soft-thresholding transform:
\[\label{l1ball_type_prior}
\beta \sim \pi^{\beta}_0, \qquad \theta = \text{sign}(\beta)\circ (|\beta|-\kappa)_+,
\]
where the operation is conducted element-wise in the second term, and 
$(\cdot)_+$ is the truncation function to $[0,\infty)$. In the projection to L1-ball with L1-norm $\theta\in\mathbb{R}^p:\|\theta\|_1\le r$ for some $r>0$, $\kappa>0$ is a scalar output from a deterministic transformation of $(\beta,r)$, as shown in \citep{xu2020bayesian}. On the other hand, one could generalize $\kappa$ to be another free non-negative parameter or a vector of non-negative values. Therefore, we refer to the generalized class of \eqref{l1ball_type_prior} as the L1-ball-type priors or models. As one motivating example (that we will revisit later), the soft-thresholded Gaussian process prior \citep{kang2018scalar} is built with $\beta$ from a Gaussian process and $\kappa$ as a scalar. This prior enjoys a nice property that $\theta$ is continuous over the input domain with high sparsity at the same time, hence is particularly useful for selecting sub-regions of an image for regression or smoothing.

Focusing on the computational aspect, the soft-thresholding transform is differentiable almost everywhere with respect to $\pi_0^{\beta}$. This means we can use off-the-shelf gradient-based MCMC algorithms \citep{duane1987hybrid,girolami2011riemann,hoffman2014no,livingstone2022barker} for its posterior estimation. The strength of these algorithms, besides low implementation cost due to the good accessibility of software \citep{carpenter2017stan,bingham2019pyro}, is in the rapid convergence to the region near the posterior mode and a high acceptance rate for changing the zero/non-zero status of multiple elements at the same time. On the other hand, notice that if the likelihood is parameterized via $\theta$ and $\kappa$ but not dependent on $\beta$ a priori, then at the state with some $\theta_j=0$, the partial derivative of the log-posterior density with respect to $\beta_j$ is zero. As a consequence, the algorithm relying on one-step diffusion [such as Metropolis-adjusted Langevin algorithm (MALA) \citep{rossky1978brownian}] would be not efficient to explore changing $\theta_j$ to a non-zero state. That is why multiple-step diffusion algorithms such as Hamiltonian Monte Carlo \citep{neal2011mcmc} and No-U-Turn sampler \citep{hoffman2014no} are preferable.

Since multiple evaluations of the gradient can be expensive, there is room for improvement that can be made in terms of gaining computational efficiency. This motivates us to do exploration in the class of Gibbs samplers, which tends to have a very low computing cost per iteration and requires almost no tuning. To avoid one-element-at-a-time Gibbs sampling, we introduce an augmented latent variable that allows us to cancel out the interaction term among the parameter elements. This enables us to update the elements in a block in each iteration. We show that this new algorithm achieves much higher efficiency in terms of effective sample size per time unit. Further, we establish the geometric ergodicity guarantee in linear regression. Interestingly, our proposed latent variable is broadly applicable beyond the L1-ball-type model setting and can be applicable in the challenging cases of estimating latent Gaussian models, such as those involving multivariate truncated Gaussian or latent Gaussian process. The source code is available on \url{https://github.com/YuZh98/Anti-correlation-Gaussian}.

\section{Blocked Gibbs Sampling with Anti-correlation Latent Gaussian}
\subsection{Motivating Problems}
We now introduce the motivating sampling problem and provide the necessary notations. Let $\theta \in\mathbb{R}^p$ be the parameter of interest, associated with precursor random variable $\beta \in\mathbb{R}^p$ via \eqref{l1ball_type_prior}. We use $\mathcal Y$ to denote the data, where $\mathcal Y_i = (x_i,y_i)$ with $y_i$ as outcome and $x_i$ as predictor. With some loss of generality, we focus on the following form on the conditional posterior:
\[\label{eq:posterior_form}
& \Pi(\theta,\beta \mid M, \phi,H,\psi,\kappa, \mathcal Y) \propto \exp \bigg[ -\frac{1}{2} (\theta'M\theta - 2 \phi'\theta)\bigg]
 \exp \bigg[ - \frac{1}{2}(\beta' H \beta- 2\psi'\beta) \bigg],\\
&   \theta = \text{sign}(\beta) \circ(|\beta|-\kappa)_+,
\]
where $M$ and $H$ are both positive semi-definite matrix of size $p\times p$, $\phi\in\mathbb{R}^p, \psi\in\mathbb{R}^p$, and $\kappa\in \mathbb{R}^p$. We treat $\kappa$ as given for now, but will discuss the case when $\kappa$ is assigned with another prior.

Although the first line in \eqref{eq:posterior_form} resembles a product of two latent Gaussian densities, the equality constraint in the second line makes it a degenerate density. A formal notation for the degenerate density can be obtained using infinitesimal
\(
p(d\theta,d\beta\mid H,\kappa,\psi)=p(\beta\mid H,\psi)\delta_{T_\kappa(\beta)}(d\theta)d\beta,\qquad T_{\kappa}(\beta)= \text{sign}(\beta) \circ(|\beta|-\kappa)_+,
\)
where $\delta$ is the Dirac measure. On the other hand, for the ease of notation, we will stick to the simplified form (\ref{eq:posterior_form}) in this article. We now provide three concrete examples to illustrate the generality of the above types of problems.

In a linear regression model for $y= X\theta +\epsilon$, with a predictor matrix $X\in\mathbb{R}^{n\times p}$, $\epsilon\sim \text{N}(0,\Omega^{-1})$ for some positive definite $\Omega$, and independent Gaussian precursor $\beta_j\stackrel{indep}\sim \text{N}(0, \tau_j)$ for some given hyper-parameter $\tau_j>0$. We have \eqref{eq:posterior_form} with $M= X'\Omega X$, $\phi = X'\Omega y$, $H=\text{diag}(1/\tau_j)$ and $\psi=0$.

In a logistic regression model with binary $y_i \sim \text{Bernoulli}[1/(1+\exp(- x_i'\theta)]$, we can use the Polya-Gamma data augmentation to create a Gaussian augmented likelihood (see \citealt{polson2013bayesian} for detail). Under independent Gaussian precursor $\beta_j\stackrel{indep}\sim \text{N}(0, \tau_j)$, we have \eqref{eq:posterior_form} with $M= X'\Omega X$, $\phi = X'\Omega (y-1/2)$, $H=\text{diag}(1/\tau_j)$, $\psi=0$, and $\Omega= \text{diag}(\omega_j)$ with $\omega_j$ marginally from the Polya-Gamma distribution $\text{PG}(1,0).$

In a sparse smoothing model $y= \theta +\epsilon$, with $\epsilon\sim \text{N}(0,I{\sigma^2})$, $y_i$ associated with some spatial coordinate $s_i$, for which the goal is to obtain a spatially correlated $\theta_i$, while making some $\theta_i$'s zero in some contiguous regions of $s_i$'s. We can use a spatially correlated Gaussian precursor $\beta\sim \text{N}[0, K(s,s) ]$ with $K$ some covariance kernel such as the squared exponential kernel. We have \eqref{eq:posterior_form} with $M= I \sigma^{-2}$, $\phi = y\sigma^{-2}$, $H=[K(s,s)]^{-1}$ and $\psi=0$.

For any non-diagonal $M$ or $H$, the quadratic term $\theta'M\theta$ or $\beta' H\beta$ in the exponent of \eqref{eq:posterior_form} makes it difficult to explore a large change in the parameter. This motivates us to use some latent variables to cancel out those terms.

\subsection{Anti-correlation Gaussian}
If $M$ is non-diagonal and we want to cancel out  $\theta'M\theta$, consider a Gaussian latent variable $r\in \mathbb{R}^p$:
\[\label{eq:anti-quad-r}
(r\mid \theta,M)\sim \No \big [ (d I -M) \theta,\;\; (d I -M) \big],
\]
where $d>\lambda_p(M)$ is a chosen constant to make $d I -M$ positive definite, with $\lambda_p(\cdot)$ the spectral norm. Similarly, if $H$ is non-diagonal, we use 
\[\label{eq:anti-quad-t}
(t \mid \beta,H)\sim \No [(e I -H) \beta,\;\; (e I -H)],
\]
with $e>\lambda_p(H)$. %

 We refer to \eqref{eq:anti-quad-r} as an ``anti-correlation Gaussian'', because it cancels out the quadratic correlation terms. For brevity, from now on, we focus on the case when we use both $r$ and $t$, and we can see that the conditional  posterior of $(\beta_j,\theta_j)$ becomes independent over $j$:
 \[\label{eq:posterior}
& \Pi(\beta,\theta \mid M, \phi,H, r,t, \kappa, \mathcal Y ) \propto \prod_{j=1}^p \exp\bigg \{ -\frac{1}{2} [ d\theta^2_j - 2 (\phi_j+r_j)\theta_j+  e \beta^2_j  -2(\psi_j+t_j)\beta_j ] \bigg\},\\
& \theta_j = \text{sign}(\beta_j) (|\beta_j|-\kappa_j)_+.
\]
The conditional independence allows us to draw $\beta_j$'s in a block. It is not hard to see that each $\beta_j$ follows a three-component mixture, corresponding to the cases when $\theta_j>0$, $\theta_j=0$ or $\theta_j<0$, for which we take a discrete variable $b_j$ taking value $1,0$ or $-1$, respectively. We have
$b_j$ in the following:
\[\label{eq:b}
& \Pi(b_j=0 \mid \cdot) = c_j^{-1}  m^{-1}_0  \exp\bigg [\frac{1}{2} \frac{(\psi_j+t_j)^2}{e}\bigg] ,\\
& \Pi(b_j=1\mid \cdot) = c_j^{-1}  m^{-1}_1  \exp\bigg [\frac{1}{2} \frac{(\phi_j + r_j +\psi_j+t_j +d \kappa_j)^2 }{d+e}- \frac{1}{2}  d\kappa^2_j  - (\phi_j+r_j)\kappa_j\bigg]
\\
& \Pi(b_j=-1\mid \cdot) = c_j^{-1} m^{-1}_{-1}   \exp\bigg [\frac{1}{2}\frac{(\phi_j + r_j +\psi_j+t_j -d \kappa_j)^2 }{d+e}- \frac{1}{2}  d\kappa^2_j  + (\phi_j+r_j)\kappa_j\bigg]
,
\]
where we use $\Pi(b\mid \cdot)$ as a shorthand notation for the full conditional posterior of $b$. In the above, $c_j^{-1}$ is a normalizing constant to make the three probabilities add up to one; $m_0$, $m_1$ and $m_{-1}$ are the constants in the following truncated Gaussian densities \eqref{eq:theta_b} that multiplied to the exponential terms:
\[\label{eq:theta_b}
& (\beta_j\mid b_j=0,\cdot)  \sim \No_{(-\kappa_j,\kappa_j)} ( \frac{\psi_j+t_j}{e}, \frac{1}{e} ),\\
& (\beta_j\mid b_j=1,\cdot)  \sim \No_{(\kappa_j,\infty)} (\frac{\phi_j + r_j +\psi_j+t_j +d \kappa_j }{d+e},\frac{1 }{d+e}),\\
& (\beta_j\mid b_j= -1,\cdot)  \sim \No_{(-\infty,-\kappa_j)} (\frac{\phi_j + r_j +\psi_j+t_j-d \kappa_j }{d+e},\frac{1 }{d+e}),
\]
where the subscript in $\No_{(\tilde \alpha_1,\tilde \alpha_2)}$ denotes the support of the truncated Gaussian.
Therefore, we can draw each $b_j$ first, and then $\beta_j$ from the corresponding truncated Gaussian. 
Note that $\theta$ is completely determined once $\beta$ and $\kappa$ are given. As a result, besides the other parameters, we have a two-block update scheme based on sampling $(r,t \mid \beta, \cdot)$, and sampling $(\beta \mid r,t, \cdot)$. See Algorithm \ref{alg:anticorr} below.

\begin{remark}
In the above, we present the simple case when each $\kappa_j$ is given. How to estimate $\kappa_j$ in MCMC depends on the structure and prior one imposes on the vector $\kappa$. In the examples of this article, we use a common $\kappa_j=\kappa_0$ for all $j$, and prior $\pi_0(\kappa_0)$. We use slice sampling step to update $\kappa_0$, with $\Pi(\kappa_0 \mid \cdot)\propto\pi_0(\kappa_0)\Pi(\theta,\beta \mid M, \phi,H,\psi, \kappa_0, \mathcal Y)$ where the second term is from  \eqref{eq:posterior}. We sample the other parameters (such as the variance of measurement error) from their respective full conditional distribution.
\end{remark}

{Empirically, we find the values of $d>\lambda_p(M)$ and $e>\lambda_p(H)$  have almost no effects on the mixing performance (Appendix \ref{mixing_de}), as long as the values do not lead to numerical overflow. In this article, since our cost of computing upper bounds $\tilde\lambda_p(M)\ge \lambda_p(M)$ or $\tilde\lambda_p(H)\ge \lambda_p(H)$ is negligibly low in each iteration of MCMC (as described in the next section), we set $d=\tilde\lambda_p(M)+\varepsilon$, $e=\tilde\lambda_p(H)+\varepsilon$  , and $\varepsilon$ is set to $10^{-6}$. For general $M$ and $H$, one can use Frobenius norm as an upper bound for $\lambda_p(M)$ or $\lambda_p(H)$.
}

\begin{algorithm}
\caption{Anti-correlation Blocked Gibbs Sampler}
\label{alg:anticorr}
\For{$k \leftarrow 1$ \KwTo $K$}{
  Sample $(r^k,t^k)$ from (\ref{eq:anti-quad-r}) and (\ref{eq:anti-quad-t})\;
  Sample $\beta^k$ from (\ref{eq:b}) and (\ref{eq:theta_b})\;
  Sample $\kappa^k$ from the full conditional distribution\;
  Compute $\theta^k$ via the second line of (\ref{eq:posterior})\;
  Sample the other parameters (if needed) from the respective full conditional distribution.
}
\end{algorithm}

\subsection{Efficient Sampling of the Anti-Correlation Gaussian in Regression}

We now focus on making the sampling of the anti-correlation Gaussian more efficient for large $p$. Motivated by regression settings and with some loss of generality, we focus on the scenario where $M$ is updated as a three-matrix product $M=X'\Omega X$, where $X$ is fixed but $\Omega$ can change in each iteration. For simplicity, we focus on sampling of $r \sim \No  [ (d I -M) \theta,\; (d I -M) ]$, and the method is trivially extensible to the sampling of $t$, if needed. 

Canonically, sampling this multivariate Gaussian involves decomposing the covariance matrix $(dI-X'\Omega X)$,  which is an $\mathcal O(p^3)$ operation. In the case of simple regression with homoscedasticity $\Omega =I \sigma^{-2}$ (such as linear regression with homoscedasticity, or the latent Gaussian in probit regression with $\sigma^2=1$, \citealt{albert1993bayesian}), we could use $d= (\lambda_p(X'X)+\tilde \epsilon)\sigma^{-2}$ with a small $\tilde \epsilon>0$, then we only need to decompose  $(d\sigma^2 I -X'X)=\mathcal L \mathcal L'$ for one time before the start of the Markov chain. In each Markov chain iteration, we can sample $r$ efficiently via $r= \sigma^{-1}\mathcal L \gamma_1 +(d I -M) \theta$ with $\gamma_1\sim \text{N}(0, I_p)$.

On the other hand, when $\Omega$ is more complicated than $I\sigma^{-2}$, and is updated in every iteration (such as in logistic regression where $\Omega$ is updated from a Polya-Gamma distribution), we would need to compute the decomposition each time. Similar issues have been noted by \cite{bhattacharya2016fast} and \cite{nishimura2022prior} during the sampling of a Gaussian with covariance $(X'\Omega X+ D)^{-1}$ with $D$ positive definite, the first group of authors exploits the Woodbury matrix identity and involves an inversion of an $n\times n$ matrix, and the second group uses a preconditioned conjugate gradient method to iteratively solve a linear system. Here our problem is slightly different, in the form of sampling of a Gaussian with covariance $(dI-X'\Omega X)$ (or equivalently, a transform $r= ( dI-X'\Omega X )\gamma_2$, with  $\gamma_2\sim \text{N}[\theta, (dI-X'\Omega X)^{-1}]$). The subtraction in the covariance/precision matrix creates a new challenge, for which we develop a new non-iterative algorithm.

For ease of presentation, we now focus on the $p\geq n$ case, and one can modify it to accommodate the $p\leq n$ case easily. Before the start of the Markov chain, we pre-compute the singular value decomposition (SVD) of $X=U_X \Lambda_X V'_X$, with $U_X$ an $n\times n$ matrix, $V_X$ an $p\times n$ matrix, $\Lambda_X$ an $n\times n$ diagonal matrix with diagonal entries $\Lambda_{X,(i,i)}\ge 0$ for $i=1,\ldots,n$. We denote $\Lambda^*_X := \Lambda_{X,(1,1)}$ as the largest singular value, and $b_\Omega:= [\lambda_n(\Omega)]^{-1}$ with $\lambda_n(\Omega)$ the largest eigenvalue of $\Omega$. We assume $d > \lambda_n(\Omega) (\Lambda^*_X )^2$, which is sufficient to ensure positive definiteness of $dI-X'\Omega X$.
Further, we obtain a matrix $V_X^\dagger\in\mathbb R^{p\times(p-n)}$ such that $[V_X \; V^{\dagger}_X]$ forms an orthonormal matrix, using the full SVD of $X$. 

\begin{remark}
The SVD is computationally expensive for large $n$ or $p$. Fortunately, we only need to run SVD for one time before starting MCMC.
The following sampler for $r$ does not involve any matrix decomposition or inversion when running MCMC. With vectorized computation on matrix-vector products, each MCMC iteration has a parallel run time $\mathcal O[\max(n,p)]$.
\end{remark}

We use the following sampling scheme.
\begin{enumerate}
        \item Sample $\gamma_1 \sim \No(0, d I_n )$, \quad$\gamma_2\sim\No(0,dI_{p-n})$, $\quad\gamma_3\sim \No [ \Lambda_X\gamma_1/d ,\;b_\Omega I_n - (\Lambda_X)^2/d ]$;
        \item Sample $\eta \sim \No(0, \Omega^{-1}- b_\Omega I_n)$;
        \item Set $r = V_X\gamma_1 + V^{\dagger}_X\gamma_2 - X'\Omega(U_X\gamma_3+ \eta) + (dI -X'\Omega X)\theta$.
\end{enumerate}
In the above, the Gaussian could be degenerate with zero variance for some elements.
\begin{theorem}
        The above algorithm produces a random sample from the anti-correlation Gaussian, $r\sim\No \big [ (dI -X'\Omega X)\theta,\;\; (dI-X'\Omega X) \big ].$
\end{theorem}

\begin{proof}
It is not hard to see that
\(
\begin{bmatrix} \gamma_1 \\ \gamma_2\\ \gamma_3\end{bmatrix} \sim
\text{N}( \begin{bmatrix} 0 \\ 0 \\ 0\end{bmatrix} ,
\begin{bmatrix} d I_n & O &  \Lambda_X \\
O& d I_{p-n} &  O\\
\Lambda_X & O  & b_\Omega I_n\end{bmatrix} ), \;\text{and}\;
\begin{bmatrix}  V_X\gamma_1+V^{\dagger}_X\gamma_2\\ U_X\gamma_3+ \eta \end{bmatrix} \sim
\text{N}( \begin{bmatrix} 0 \\ 0 \end{bmatrix} ,
\begin{bmatrix} dI_p & X' \\ X  & \Omega^{-1}\end{bmatrix} ),
\)
where $O$ in the above denotes a conformable matrix filled with zeros. The first covariance is positive definite because $d b_\Omega > \Lambda^2_{X,i,i}$ for any $i=1,\ldots,n$.
Then we see that $\text{VAR}[ V_X\gamma_1 + V^{\dagger}_X\gamma_2  - X'\Omega(U_X\gamma_3+ \eta)] = dI+ X'\Omega\Omega^{-1}\Omega X- 2 X'\Omega X= dI-X'\Omega X$. 
Adding $(dI- X'\Omega X)\theta$ yields the result.
\end{proof}

We now discuss the complexity of the above sampling algorithm. The sampling of $\gamma_1$, $\gamma_2$, and $\gamma_3$ can be carried out very efficiently due to the diagonal covariance. Often $\Omega$ is a diagonal matrix, hence sampling of $\eta$ can be quite efficient as well.
For the other cases with non-diagonal $\Omega$ at a large $n$, one often adopts a structured (such as close-to-low-rank) $\Omega$ or $\Omega^{-1}$ under which $\eta$ can be sampled efficiently \citep{ghosh2009default,chandra2021bayesian}.

\section{Geometric Ergodicity Guarantee}
{
Now we establish the convergence guarantee of our proposed  Gibbs sampling algorithm for linear models. For mathematical tractability, we focus on the models with  fixed$(M, \phi, H,\psi,\kappa)$, hence the two-block updating scheme. Since $\kappa$ is assumed fixed, $\theta$ is uniquely determined by $\beta$. In order to avoid the complication from degenerate distribution, we only focus on $\beta$ and $(r,t)$ in the context of the distribution or measure.

Our target distribution is associated with the joint posterior probability kernel function for $(\beta,r,t)$:
\(
\Pi(r,t \mid M, H,\beta)\Pi(\beta \mid M, \phi,H,\psi,\kappa, \mathcal Y),
\)
and we denote the associated measure by $\mu_{\beta,r,t}(\cdot)$. In MCMC, we start with initial values for those latent variables drawn from a certain initial distribution, and we denote them by $(r^0,t^0)$, then we update their values iterative via the following Markov transition kernel:
\(
	& \mathcal K(\beta^{m+1},r^{m+1}, t^{m+1} \mid r^m,t^m ) \\
	& =  \Pi_{\mathcal K}(r^{m+1}, t^{m+1}\mid \beta^{m+1}) \Pi_{\mathcal K}(\beta^{m+1}\mid  r^{m}, t^{m})
\)
for $m=0,1,2,\ldots$, where the first two densities are based on \eqref{eq:anti-quad-r} and that of $t$, and the last one is based on \eqref{eq:b} and \eqref{eq:theta_b}, with all the fixed parameters omitted. We denote the associated measure for $(\beta^{m},r^{m},t^{m})$ by $P_{\beta,r,t}^m[(r^0,t^0),\cdot]$.

In order to show that $P_{\beta,r,t}^m[(\beta^0,r^0,t^0),\cdot]$ converges to $\mu_{\beta,r,t}(\cdot)$, we can simply show that $P^m_{r,t}[(r^0,t^0),\cdot]$ converges to $\mu_{r,t}(\cdot)$, with $\mu_{r,t}(\cdot)$ the marginal posterior measure integrated over $\beta$, and $P^m_{r,t}[(r^0,t^0),\cdot]$ the marginal measure of $(r^m,t^m)$ after $m$ MCMC iterations after integrating out $\beta^m$. A similar technique was used by \cite{hobert2012tvgs} for a two-part update scheme. With those ingredients in place, we are ready to state the geometric ergodicity results.
}

\begin{theorem}\label{lemma:rt}
    There exists a real-valued function $C_1(r^0,t^0)$ and $0<\gamma<1$ such that for all $(r,t)$,
    \[\label{eq:marginal_ge}
    \|P^m_{r,t}[(r^0,t^0),\cdot]- \mu_{r,t}(\cdot)\|_{\text{TV}}\leq C_1(r^0,t^0)\gamma^m,
    \]
    \noindent where $\|\cdot\|_{TV}$ denotes the total variation norm.
\end{theorem}
  \noindent We defer the proof to Appendix \ref{geometricergodicityproof}. Further, we have the following result.

\begin{theorem}
    The Markov chain generated by $\mathcal K(\beta^{m+1},r^{m+1}, t^{m+1} \mid \beta^{m}, r^m,t^m )$ is geometrically ergodic. Specifically, there exists a real-valued function $C_2(r^0,t^0)$ and $0<\gamma<1$ such that for all $(r^0,t^0)$,
    \[
    \|P^m_{(\beta,r,t)}[(r^0,t^0),\cdot]-\mu_{(\beta,r,t)}(\cdot)\|_{\text{TV}}\leq C_2(r^0,t^0)\gamma^m,
    \]
    where $\gamma$ is the same rate as the one in \eqref{eq:marginal_ge}.
\end{theorem}

The proof is a direct extension of Theorem \ref{lemma:rt}, since $\Pi_{\mathcal K}(\beta \mid r,t)$ in the Markov transition kernel after $m$ iterations and $\Pi(\beta \mid r,t, M, \phi,H,\psi,\kappa, \mathcal Y)$ in the posterior distribution are the same thing. Therefore, the convergence rate remains the same. See \cite{DiaconisKK2008GibbsSE}; \cite{juns1994covariance}; \cite{robert1995convergencecontrol}; \cite{rosenthal2001deinitializing} for detail. 

\begin{remark}
To be rigorous, we obtain the above convergence theory for the two-block data augmentation sampler. For more complicated samplers that involve more than two blocks (such as those in our numerical examples), we observe fast empirical convergence and hence expect the theory to be generalizable.
\end{remark}

\begin{remark}
Our convergence result is only qualitative, and getting a rate quantification of $\gamma$ is challenging and beyond the scope of this article. We refer the readers to \cite{10.1214/19-EJS1563} for useful analytical techniques that could lead to a tractable rate calculation. In Appendix B.2, we provide numerical evidence of fast convergence for dimensions up to $p=5000$.
\end{remark}

On computational efficiency, a data augmentation-based Gibbs sampler, which alternates in drawing from $\Pi(\theta\mid r, y,\cdot)$ and $\Pi(r\mid \theta, y,\cdot)$ in closed forms, often has a slower mixing rate than a marginal sampler targeting $\Pi(\theta\mid  y,\cdot)= \int \Pi(\theta, r\mid  y,\cdot) \textup d r$. However, when there lacks a closed form for drawing from the marginal $\Pi(\theta\mid  y,\cdot)$, one has to rely on computationally expensive Markov transition kernel (such as No-U-Turn leapfrog steps) to update $\theta$. Therefore, there is a balance to be struck --- in order to produce a fixed size $m$ effective samples of $\theta$'s, one either runs (i) a marginal density-based algorithm that is relatively slow in each iteration but with fast mixing (hence fewer iterations needed), or (ii)  a data augmentation-based algorithm that is fast in each iteration but with relatively slow mixing (hence more iterations needed). Clearly, one would prefer the algorithm that takes a shorter time to produce $m$ effective samples, as measured by a higher effective sample size per time unit (ESS/time). We provide empirical evidence in Section 5.1 that our data augmentation-based sampler has a higher ESS/time than some popular marginal density-based sampler.

\section{Extensions}
The data augmentation using anti-correlation Gaussian can be employed beyond the scope of L1-ball-type models. We now discuss two useful classes of extensions.

The first extension is on the sampling of latent Gaussian models \citep{robert_monte_2004,fahrmeir_multivariate_2010,gelman_bayesian_2014}, whose likelihood generally takes the form:
\[
L(y;\theta) \propto \int  |{M}({\theta})|^{1 / 2} \exp \left[-\frac{1}{2} z^T M({\theta})z \right ]\prod_{i=1}^n  g\left(y_i \mid z_i, \theta\right) dz,
\]
with $g$ some density/mass function for the data $y$, and $M(\theta)$ is a positive definite matrix that depends on $\theta$ as the parameter of interest. This framework covers a wide range of applications such as spatial, regression, or dynamic modeling. Although it seems intuitive to simply sample $z$ as a latent variable during MCMC, the presence of non-Gaussian $g$ often disrupts posterior conjugacy, creating difficulty in efficiently updating $z$ especially when $z$ is in high dimension. This is one major motivation for posterior approximation methods such as integrated nested Laplace approximation (INLA) \citep{rue2009approximate} as an alternative to MCMC.

Using the anti-correlation Gaussian $(r \mid z,\theta)\sim N \{ [dI- M(\theta)]z, dI- M(\theta)\}$, we can obtain conditionally independent $z_i$ that can be updated efficiently in a block. 
\[
\Pi(z\mid r,\theta, y) \propto  \prod_{i=1}^n g\left(y_i \mid z_i, \theta\right) \exp( - \frac{1}{2} d z_i^2 + z_ir_i)
\]
Due to the algorithmic similarity to the sampling of the soft-thresholded Gaussian process model to be presented in the numerical experiments, we omit further detail on the sampling for latent Gaussian models.

The second extension is on the sampling of truncated multivariate Gaussian:
\[\label{pdf:truncMVN}
\Pi(\theta \mid \mu,\Sigma, R) \propto  \exp \left [  -\frac{1}{2} (\theta-\mu)^{\rm T} \Sigma^{-1} (\theta-\mu) \right ] 1(\theta\in R)
\]
where $\mu\in \mathbb{R}^p$, $\Sigma$ is positive definite, and $R$ some constrained set of dimension $p$. Despite being commonality in statistical application, the truncated multivariate Gaussian can pose a challenge for sampling, especially in high dimension \citep{Wilhelm2010tmvtnormAP,chow2014preconditioned}, due to having correlation and truncation at the same time. Using the anti-correlation Gaussian $(r \mid \theta,\mu, \Sigma)\sim N [ (dI- \Sigma^{-1}) (\theta-\mu), dI- \Sigma^{-1}]$, we have 
\[
\Pi(\theta \mid \mu,r ) \propto 1(\theta\in R) \prod_{j=1}^p\exp[ - \frac{1}{2} d(\theta_j-\mu_j)^2 + (\theta_j-\mu_j)r_j ].
\]
When the constraints in $R$ are separable over each sub-dimension, then $\theta_j$ is conditionally independent over $j$. For example, with simple box constraint $R=\cap_{j=1}^p\{x: l_j<x_j<r_j\}$, then each $\theta_j$ follows a univariate truncated Gaussian and $\theta$ can be sampled in a block. We give a numerical example in Appendix \ref{sec:truncMVN}.

\section{Numerical Examples}
We now use several numerical examples to show the computational efficiency of the anti-correlation blocked Gibbs sampler.

\subsection{Linear Regression}\label{sec:lin_reg}

 We first consider the setting where $M$ is non-diagonal. We use
 linear regression with $y_i\sim\text{N}(x_i^{\rm T}\theta,\sigma^2),$ $x_i\in \mathbb{R}^{p}$ simulated from a multivariate Gaussian with mean zero, and correlation $\rho^{|j-j'|}$ between $x_{i,j}$ and $x_{i,j'}$. During data generation, in each experiment, we set the ground truth $\sigma^2=1$ and $\theta$ as one of the vectors from a range of signal-to-noise ratios, as previously used by
  \cite{yang2016computational}. Specifically, we set the first 10 elements of $\theta$ to
 \(
 c\sqrt{\frac{\sigma^2\log(p)}{n}}(2,-3,2,2,-3,3,-2,3,-2,3)^T,
 \)
 where $c$ is the selected signal-to-noise ratio, taken from $\{1,2,3,6\}$, with $c=1$ in a low-to-moderate signal regime, and $c\geq2$ in a strong signal regime \citep{yang2016computational}. The other elements of $\theta$ are set to zero. For the prior specification, we use $\beta_j\stackrel{indep}\sim \text{N}(0, \tau_j)$, with $\tau_j\sim \text{Inverse-Gamma}(a_j,b_j)$, and $\kappa_0\sim \text{Exp}(\lambda)$ (we use rate parameterization for all inverse-gamma and exponential distributions). We choose the exponential prior for $\lambda$ as a weakly-informative prior, so that the $\kappa_0$ should be relatively small as a priori,  while the light tail of exponential allows the posterior to have potentially large $\kappa_0$. For the noise variance, we use $\sigma^2\sim \text{Inverse-Gamma}(a_\sigma,b_\sigma)$. In the experiments, we use $n=300$, $p\in\{10,50,500\}$, $\rho\in\{0.5,0.9\}$, $a_j=5$, $b_j=1$, $a_\sigma=1$, $b_\sigma=1$ and $\lambda=1$.

For benchmarking, we compare our algorithm against three other samplers: (i) the No-U-Turn (NUTS) sampler \citep{hoffman2014no}; (ii) the component-wise slice sampler \citep{neal2003slice} that updates one element at a time; (iii) the empirical principal component-based slice sampler \citep{thompson2011slice} with randomized directions. To provide more details about (iii), we use an empirical estimate of the covariance between all parameters (based on running 200 iterations of the component-wise slice sampler during an adaptation period, and the adaptation time cost is excluded from the calculation of computing efficiency), take the first $10$ principal components, and run the slice sampler along each direction under slight perturbation (by adding Gaussian noise $\text{N}(0,1/4p)$ to each element of direction, then rescaling the perturbed direction vector to have unit 2-norm). For (ii) and (iii), we use the stepping out and shrinkage procedures  \citep{neal2003slice} to draw a subset of slices.

We conduct 10 times of repeated experiments under each setting, with different random seed numbers. At $\rho=0.5$, we run MCMC for 10,000 iterations, and treat the first 2,000 as burn-ins; at $\rho=0.9$, since the high pairwise correlation slightly slows down the mixing, we run MCMC for 20,000 iterations, and treat the first 10,000 as burn-ins. We record the total wall clock time (in seconds), and calculate the effective sample size. Since the popular NUTS sampler implementation, Stan \citep{carpenter2017stan} is coded in C++, while all the other three algorithms are done in R, to have a fair comparison of the algorithmic efficiency,
we estimate the running time of NUTS using a pure R implementation from the `hmclearn' package \citep{thomas2021learning}. Specifically, we first obtain the total number of leap-frog steps from the diagnostics of Stan program, then run the same number of leap-frog steps in `hmclearn' and record wall clock time. 
\begin{figure}[H]
\begin{subfigure}[t]{.24\textwidth}
    \begin{overpic}[width=\textwidth]{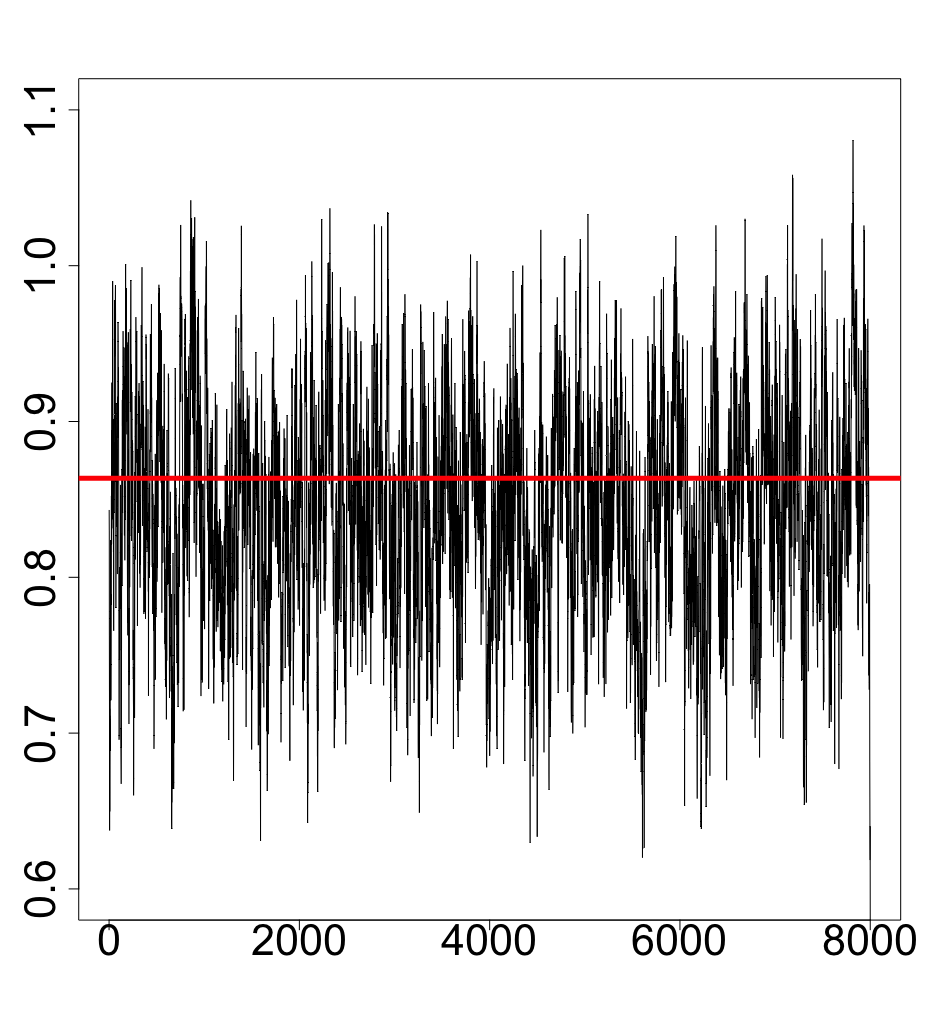}
        \put(40, -1){\scriptsize Iteration}
    \end{overpic}
    \caption{\scriptsize Trace plot of $\theta_1$  Anti-corr Gaussian.}
\end{subfigure}
\begin{subfigure}[t]{.24\textwidth}
    \begin{overpic}[width=\textwidth]{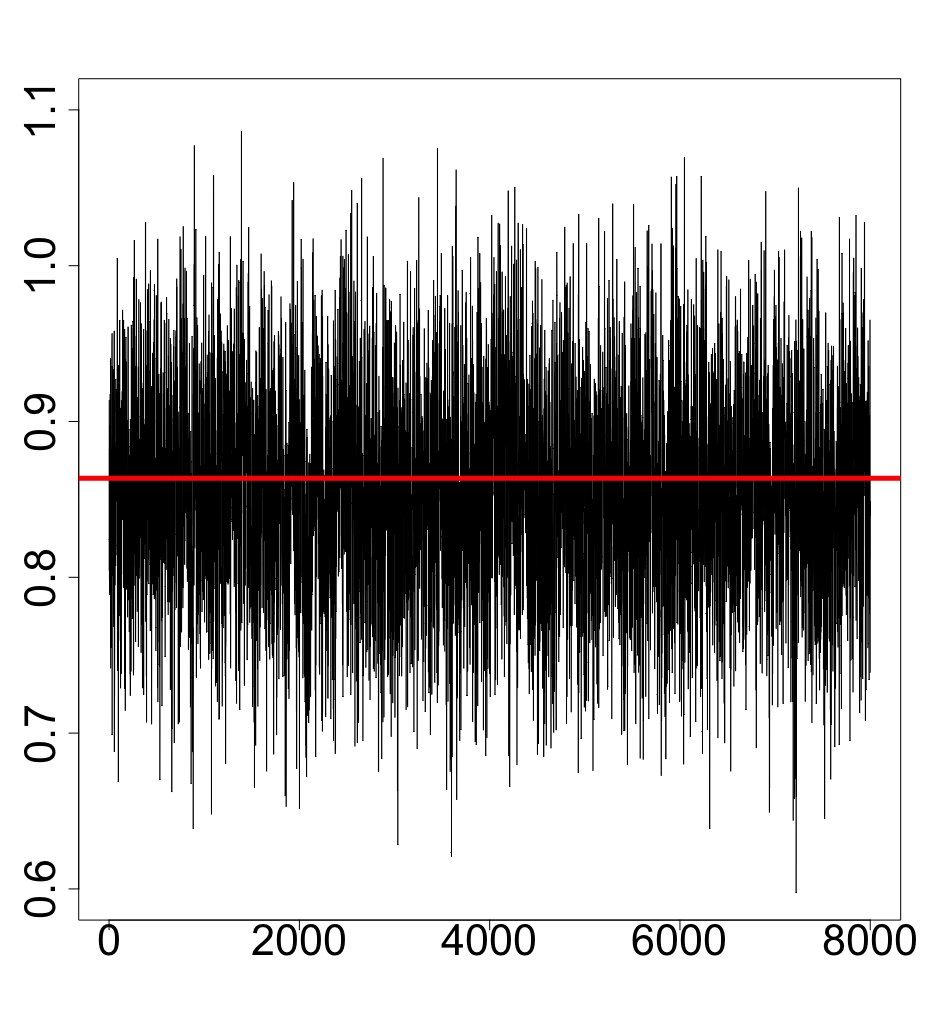}
    \put(40, -1){\scriptsize Iteration}
    \end{overpic}
    \caption{\scriptsize Trace plot of $\theta_1$ for NUTS.}
\end{subfigure}
\begin{subfigure}[t]{.24\textwidth}
    \begin{overpic}[width=\textwidth]{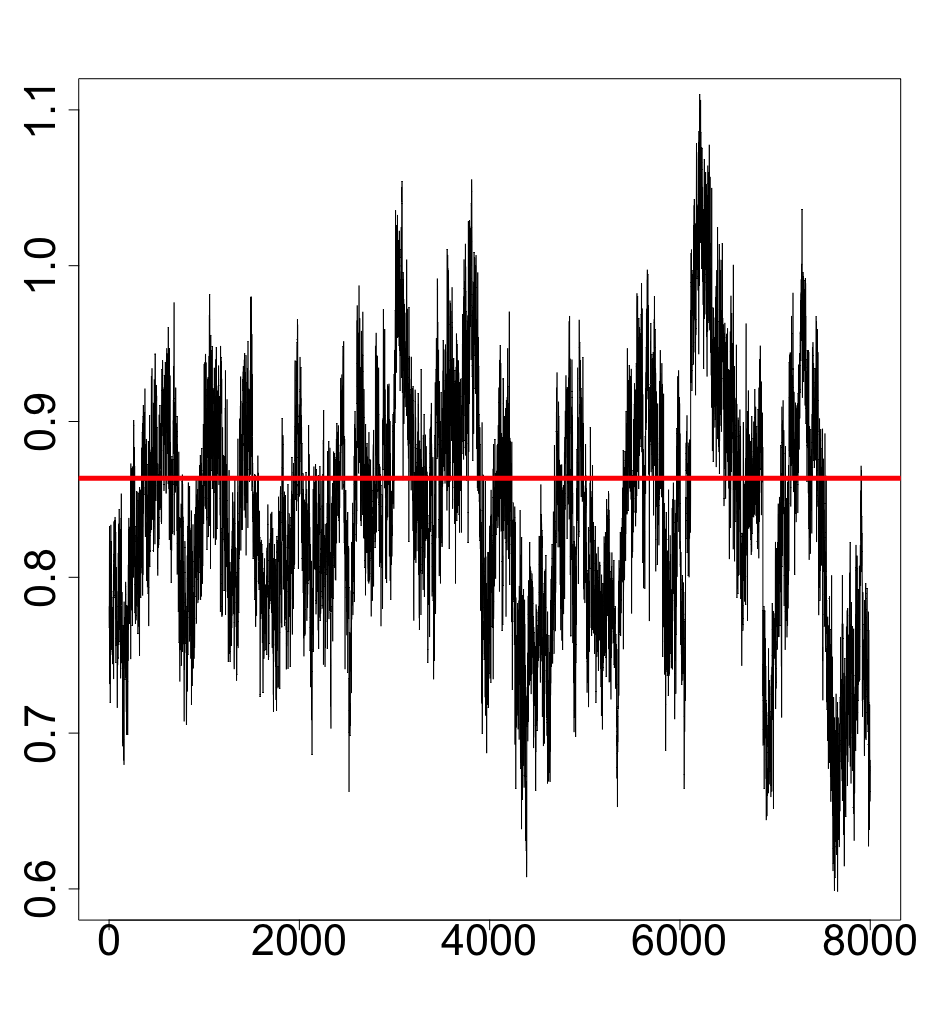}
    \put(40, -1){\scriptsize Iteration}
    \end{overpic}
    \caption{\scriptsize Trace plot of $\theta_1$ for EPC slice.}
\end{subfigure}
\begin{subfigure}[t]{.24\textwidth}
    \begin{overpic}[width=\textwidth]{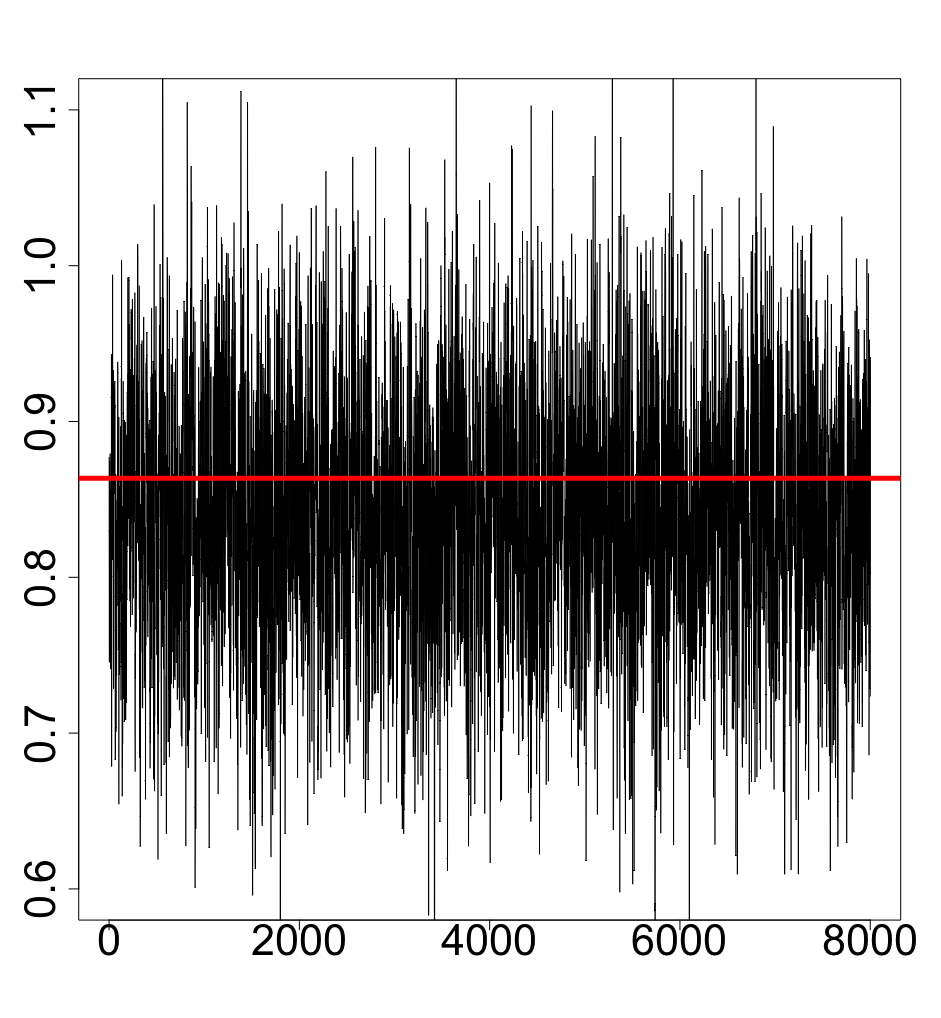}
    \put(40, -1){\scriptsize Iteration}
    \end{overpic}
    \caption{\scriptsize Trace plot of $\theta_1$ for Comp-wise slice.}
\end{subfigure}
\caption{Trace plots for the four algorithms: anti-correlation Gaussian (Anti-corr Gaussian), No-U-Turn Sampler (NUTS), empirical principal component-based slice sampler (EPC slice), and component-wise slice sampler (Comp-wise slice). The horizontal lines are the ground truth.\label{fig:tp1}}
\end{figure}

To compare the algorithms in terms of variable selection and estimation accuracy, we report the false positive rates (FPR), false negative rates (FNR) and mean squared errors (MSE) for estimating regression coefficients $\theta$ in Tables \ref{tb:FPRandFNR} and \ref{tb:MSE}. {The FPR and FNR are computed based on the element-wise $95\%$ credible intervals of $\theta_j$'s. The FPR is the proportion of the intervals corresponding to the ground truth zero that do not cover zero, while the FNR is the proportion of the intervals corresponding to the non-zero ground truth that cover zero.} When $c=1$, the signal is not strong, and all algorithms fail to produce good results. At $\rho=0.9$, the correlations within the design matrix are strong, and anti-correlation blocked Gibbs sampler outperforms all the other algorithms. At $\rho=0.5$ and $c\geq2$, the anti-correlation Gaussian, the NUTS, and the component-wise slice sampler work similarly well, while the EPC slice sampler gives bad results when $p=500$. We find similar results in MSE.

\begin{table}[H]
\centering
\begin{tabular}{|l|l|l|l|l|l|}
\hline
                   &                    & $c=1$ & $c=2$ & $c=3$ & $c=6$ \\ \hline
           &                    & FPR,FNR & FPR,FNR & FPR,FNR & FPR,FNR \\ \hline\hline
\multirow{4}{*}{$p=$10, $\rho=$0.5}  & Anti-corr Gaussian & -, 65  & -, 2   & -, 0  & -, 0   \\ \cline{2-6}
                            & NUTS               & -, 49  & -, 0   & -, 0  & -, 0  \\ \cline{2-6}
                            & Comp-wise slice    & -, 96  & -, 7   & -, 0  & -, 0  \\ \cline{2-6}
                            & EPC slice          & -, 56  & -, 6   & -, 0  & -, 0  \\ \hline
\multirow{4}{*}{$p=$50,  $\rho=$0.5}  & Anti-corr Gaussian & 0, 67  & 0, 0   & 0, 0  & 0, 0   \\ \cline{2-6}
                            & NUTS               & 0, 71  & 0, 1   & 0, 0  & 0, 0  \\ \cline{2-6}
                            & Comp-wise slice    & 0, 87  & 0, 5   & 0, 0  & 0, 0   \\ \cline{2-6}
                            & EPC slice          & 0, 79  & 0, 0   & 0, 0  & 0, 0   \\ \hline
\multirow{4}{*}{$p=$500,  $\rho=$0.5} & Anti-corr Gaussian & 0, 65  & 0, 0   & 0, 0  & 0, 0  \\ \cline{2-6}
                            & NUTS               & 0, 72  & 0, 0   & 0, 0  & 0, 0    \\ \cline{2-6}
                            & Comp-wise slice    & 0, 87  & 0, 0   & 0, 0  & 0, 0   \\ \cline{2-6}
                            & EPC slice          & 0, 95  & 0, 81  & 0, 52 & 0, 23  \\ \hline
\multirow{4}{*}{$p=$10,  $\rho=$0.9}  & Anti-corr Gaussian & -, 100 & -, 55  & -, 0  & -, 0   \\ \cline{2-6}
                            & NUTS               & -, 99  & -, 69  & -, 31 & -, 0 \\ \cline{2-6}
                            & Comp-wise slice    & -, 100 & -, 100 & -, 94 & -, 73 \\ \cline{2-6}
                            & EPC slice          & -, 98  & -, 64  & -, 21 & -, 0 \\ \hline
\multirow{4}{*}{$p=$50,  $\rho=$0.9}  & Anti-corr Gaussian & 0, 94  & 0, 74  & 0, 0  & 0, 0  \\ \cline{2-6}
                            & NUTS               & 0, 100 & 0, 87  & 0, 49 & 0, 0   \\ \cline{2-6}
                            & Comp-wise slice    & 0, 98  & 0, 95  & 0, 89 & 0, 60   \\ \cline{2-6}
                            & EPC slice          & 0, 89  & 1, 80  & 2, 39 & 3, 5   \\ \hline
\multirow{4}{*}{$p=$500,  $\rho=$0.9} & Anti-corr Gaussian & 0, 86  & 0, 77  & 0, 0  & 0, 0   \\ \cline{2-6}
                            & NUTS               & 0, 96  & 0, 85  & 0, 20 & 0, 0   \\ \cline{2-6}
                            & Comp-wise slice    & 0, 90  & 0, 90  & 0, 60 & 0, 50  \\ \cline{2-6}
                            & EPC slice          & 0, 87  & 0, 82  & 0, 81 & 0, 75  \\ \hline
\end{tabular}
\caption{False positive rates (FPR) and false negative rates (FNR) (measured in $\%$) for the four algorithms: anti-correlation Gaussian (Anti-corr Gaussian), No-U-Turn Sampler (NUTS), empirical principal component-based slice sampler (EPC slice), and component-wise slice sampler (Comp-wise slice). When $p=10$, there is no FPR as all ground-truth $\theta_j$'s are non-zero.}
\label{tb:FPRandFNR}
\end{table}

\begin{table}[H]
\centering
\begin{tabular}{|c|c|c|c|c|c|}
\hline
             &                    & $c=1$  & $c=2$  & $c=3$  & $c=6$ \\ \hline
\multirow{4}{*}{$p=$10, $\rho=$0.5}   & Anti-corr Gaussian & 0.0082 & 0.0056 & 0.0058 & 0.0063  \\ \cline{2-6}
                            & NUTS               & 0.0078 & 0.0042 & 0.0067 & 0.0059   \\ \cline{2-6}
                            & Comp-wise slice    & 0.0192 & 0.0045 & 0.0047 & 0.0067   \\ \cline{2-6}
                            & EPC slice          & 0.0118 & 0.0072 & 0.0092 & 0.0054   \\ \hline
\multirow{4}{*}{$p=$50, $\rho=$0.5}   & Anti-corr Gaussian & 0.0041 & 0.0018 & 0.0015 & 0.0017  \\ \cline{2-6}
                            & NUTS               & 0.0057 & 0.0013 & 0.0014 & 0.0013   \\ \cline{2-6}
                            & Comp-wise slice    & 0.0053 & 0.0014 & 0.0016 & 0.0016  \\ \cline{2-6}
                            & EPC slice          & 0.0101 & 0.0021 & 0.0020 & 0.0018  \\ \hline
\multirow{4}{*}{$p=$500, $\rho=$0.5}  & Anti-corr Gaussian & 0.0009 & 0.0002 & 0.0002 & 0.0001  \\ \cline{2-6}
                            & NUTS               & 0.0009 & 0.0002 & 0.0002 & 0.0001   \\ \cline{2-6}
                            & Comp-wise slice    & 0.0025 & 0.0003 & 0.0001 & 0.0003   \\ \cline{2-6}
                            & EPC slice          & 0.0026 & 0.0087 & 0.0126 & 0.0217   \\ \hline
\multirow{4}{*}{$p=$10, $\rho=$0.9}  & Anti-corr Gaussian & 0.0379 & 0.0358 & 0.0597 & 0.0810   \\ \cline{2-6}
                            & NUTS               & 0.0329 & 0.0515 & 0.0561 & 0.0628   \\ \cline{2-6}
                            & Comp-wise slice    & 0.0402 & 0.1270 & 0.2410 & 0.3436   \\ \cline{2-6}
                            & EPC slice          & 0.0351 & 0.0572 & 0.0498 & 0.0791  \\ \hline
\multirow{4}{*}{$p=$50, $\rho=$0.9}  & Anti-corr Gaussian & 0.0150 & 0.0310 & 0.0217 & 0.0216   \\ \cline{2-6}
                            & NUTS               & 0.0137 & 0.0249 & 0.0262 & 0.0183   \\ \cline{2-6}
                            & Comp-wise slice    & 0.0144 & 0.0489 & 0.0886 & 0.1316   \\ \cline{2-6}
                            & EPC slice          & 0.0140 & 0.0458 & 0.0578 & 0.0516   \\ \hline
\multirow{4}{*}{$p=$500, $\rho=$0.9} & Anti-corr Gaussian & 0.0023 & 0.0074 & 0.0014 & 0.0020  \\ \cline{2-6}
                            & NUTS               & 0.0019 & 0.0068 & 0.0009 & 0.0005  \\ \cline{2-6}
                            & Comp-wise slice    & 0.0024 & 0.0082 & 0.0128 & 0.0116   \\ \cline{2-6}
                            & EPC slice          & 0.0025 & 0.0099 & 0.0210 & 0.0742  \\ \hline
\end{tabular}
\caption{Mean squared error (MSE) for estimating $\theta$ via the four algorithms under different settings. }
\label{tb:MSE}
\end{table}

At $c=3$ and $\rho=0.5$, we report in Table \ref{tb:running_time} the average running time for 1,000 iterations. The anti-correlation Gaussian is the fastest, followed by EPC slice, component-wise slice, and NUTS. As shown in Figure \ref{fig:tp1}, the NUTS algorithm has the fastest mixing; however, this is at the cost of expensive computation per iteration. To further illustrate the advantage of the anti-correlation blocked Gibbs sampler in terms of computational efficiency, we compare the effective sample size per computing time in Table \ref{tb:ESS_s}. It can be seen that the anti-correlation Gaussian has the highest effective sample size per computing time.

\begin{table}[H]
\centering
\begin{tabular}{|c|c|c|c|c|}
\hline
$p$  & Anti-corr Gaussian & NUTS & Comp-wise slice & EPC slice \\ \hline
10   &         \textbf{0.41}           &   26.18   &       3.61          &      3.04     \\ \hline
50   &         \textbf{0.93}           &   108.33   &       31.44          &      4.16     \\ \hline
500  &         \textbf{5.60}           &   22875.74   &        1121.47         &     16.14      \\ \hline
\end{tabular}
\caption{ Running time for 1,000 iterations for the four algorithms.
The time unit is in seconds based on pure R implementation for each algorithm.
}
\label{tb:running_time}
\end{table}

\begin{table}[H]
\centering
\begin{tabular}{|c|c|c|c|c|}
\hline
($p$, $\rho$) & Anti-corr Gaussian & NUTS & Comp-wise slice & EPC slice \\ \hline
(10, 0.5)     &         \textbf{202.36, 265,56}           &  8.82, 7.05    &     49.08, 60.31            &      5.19, 8.23     \\ \hline
(50, 0.5)     &         \textbf{62.86, 165.74}           &  1.52, 2.46    &      5.12, 4.61           &      4.34, 5.98     \\ \hline
(500, 0.5)    &        \textbf{4.81, 34.95}            &  0.01, 0.01    &      0.10, 0.14           &      2.36, 7.58     \\ \hline
(10, 0.9)     &        \textbf{31.19, 35.79}            &  4.22, 3.89    &       No convergence          &     3.61, 4.48      \\ \hline
(50, 0.9)     &        \textbf{11.28, 15.18}            &   0.53, 1.05   &        No convergence       &    No convergence       \\ \hline
(500, 0.9)    &        \textbf{3.05, 14.50}            &  $<$0.01, $<$0.01    &      No convergence           &     No convergence      \\ \hline
\end{tabular}
\caption{Effective sample size per computing time (ESS/s) for the four algorithms. In each cell, the first number is the average ESS/s for the first 10 entries, and the second number is the average ESS/s for the rest entries. Each ESS/s is reported as the average result of 10 repetitions with a different seed number. The two slice samplers fail to converge to the ground truth in some settings within $20,000$ iterations.}
\label{tb:ESS_s}
\end{table}

\begin{figure}[H]
\centering
\begin{subfigure}[t]{0.75\textwidth}
    \begin{overpic}[width=\textwidth]{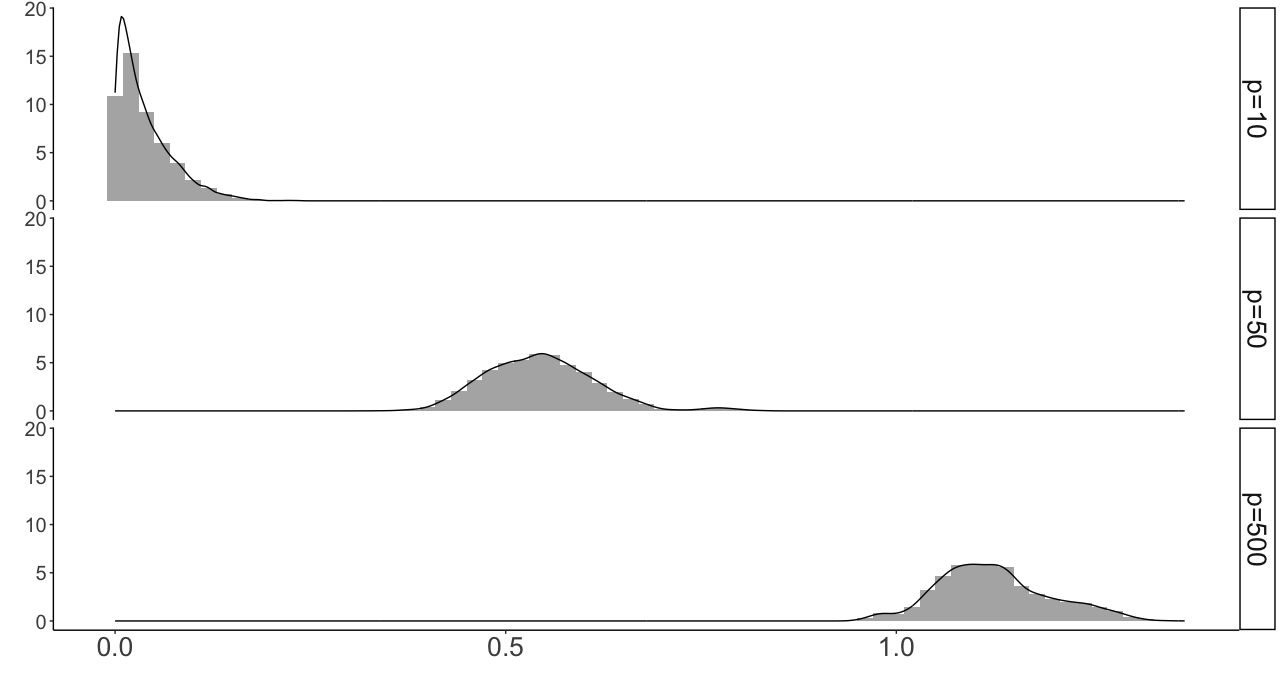}
    \put(40, -1){$\kappa_0$}
    \put(-1,25){\rotatebox{90}{\scriptsize Density}}
    \end{overpic}
\end{subfigure}
\caption{Posterior density estimations of $\kappa_0$ for varying dimensions $p$ when $(c,\rho)=(3,0.5)$.}
\label{fig:dist_kappa}
\end{figure}

Figure \ref{fig:dist_kappa} shows the posterior distribution of $\kappa_0$ under different values of $p$. As $p$ increases, the mode of the posterior density of $\kappa_0$ increases, as a result of the higher level of sparsity.

\subsection{Application on Soft-Thresholded Gaussian Process}
We next focus on the setting with non-diagonal $H$. We consider an application of the soft-thresholded Gaussian process \citep{kang2018scalar}:
\(
   \theta_s = \text{sign}(\beta_s) \circ(|\beta_s|-\kappa_0)_+,
\qquad  \beta \sim \text{GP} [0, K(\cdot, \cdot)],
\)
where $s$ is the location in some input location space $\mathbb{S}$, $\kappa_0>0$ is a scalar, and $\text{GP}$ denotes a Gaussian process from which any finite-dimensional realization follows a multivariate  (potentially degenerate) Gaussian with zero mean, and covariance parameterized by $\text{Cov}(\beta_s, \beta_{s'})=K(s,s')$. With an appropriate choice of $K$, we can obtain useful properties on a GP realization,  such as continuity or smoothness over $s$. After applying soft-thresholding, we preserve the continuity over $\mathbb{S}$, and smoothness in each open set where $\theta_s\neq 0$. Therefore, the above is useful for obtaining a ``smooth {\em and} sparse'' parameterization for $\theta$. In this article, we consider a simple but useful application in image smoothing, in the form of:
\(
y_s = \theta_s + \epsilon_s, \qquad \epsilon_s \stackrel{iid} \sim \text{N}(0,\sigma^2),
\)
where $s$ is the pixel location, $s=(i,j)$ with $i=1,\ldots, n_1$ and $j=1,\ldots, n_2$. We set the covariance function as $K(s,s')=\tau\exp[-{\|s-s'\|^2_2}/(2\xi^2)]$.

In application, we use the functional magnetic resonance imaging scan of one human subject who was performing a motor task. We take a sectional view along the anterior-posterior axis of the brain, corresponding to $91\times 91$ pixels per frame, and over a time period, we collect $280$ frames. Using superscript $f$ to index each frame, we assume $y^f =\{y_s^f\}_{\text{all }s}$ to be independent over $f$ with frame-varying mean $\theta^f=\{\theta_s^f\}_{\text{all } s}$, but the other parameters $\sigma^2$, $\tau$, $\xi$ and $\kappa_0$ are invariant over $f=1,\ldots, 280$. To finish the prior specification, we use $\tau\sim\text{Inverse-Gamma}(0.1,0.1)$, $\sigma^2\sim\text{Inverse-Gamma}(0.1,0.1)$ and $\kappa_0\sim\text{Exp}(0.5)$.  \cite{tikhonov2020joint} has previously shown that when the goal of using GP is to obtain smoothing of data, estimating the bandwidth $\xi$ under a reasonable degree of precision is empirically adequate. Specifically, they impose a simple discrete prior on $\xi$ supported on a finite set of values, so that all possible precision matrix $H$ and its determinant could be pre-computed, hence further facilitating the computation. We follow their solution here, with $\xi =0.5 \tilde t $ and $\tilde t$ a discrete uniform prior supported on $\{1,\ldots, 20\}$.

\begin{figure}[H]
\begin{subfigure}[t]{0.27\textwidth}
    \begin{overpic}[width=\textwidth]{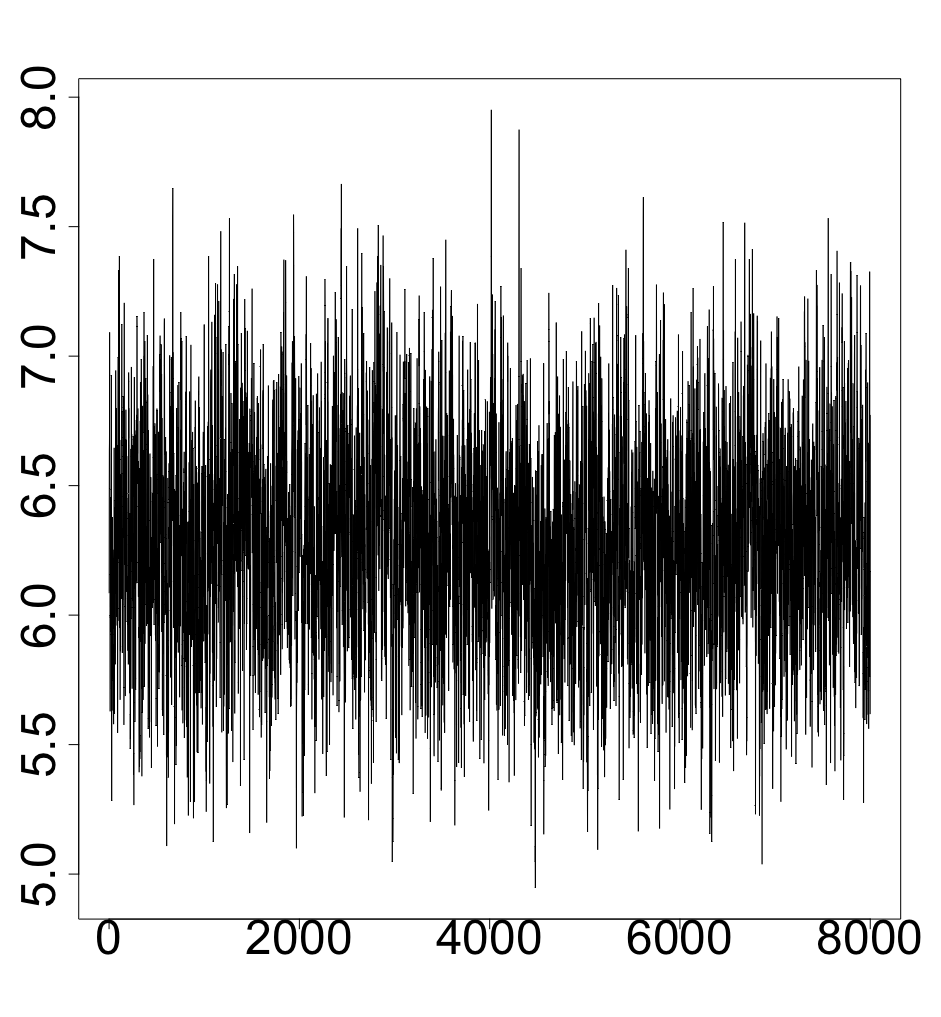}
    \put(40, -1){\scriptsize Iteration}
    \end{overpic}
    \caption{Traceplot of $\theta^{100}_{3320}$.}
\end{subfigure}
\begin{subfigure}[t]{0.27\textwidth}
    \begin{overpic}[width=\textwidth]{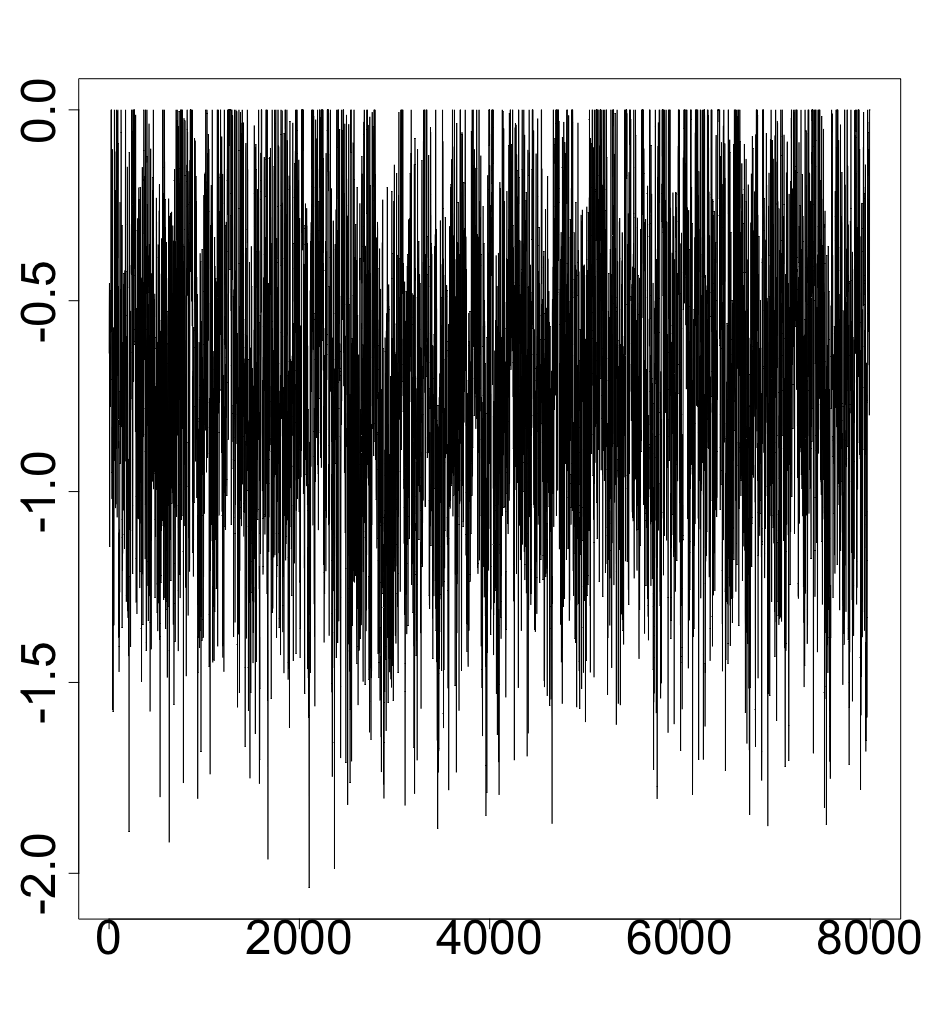}
    \put(40, -1){\scriptsize Iteration}
    \end{overpic}
    \caption{Traceplot of $\theta^{100}_{6700}$.}
\end{subfigure}
\begin{subfigure}[t]{0.41\textwidth}
    \begin{overpic}[width=\textwidth]{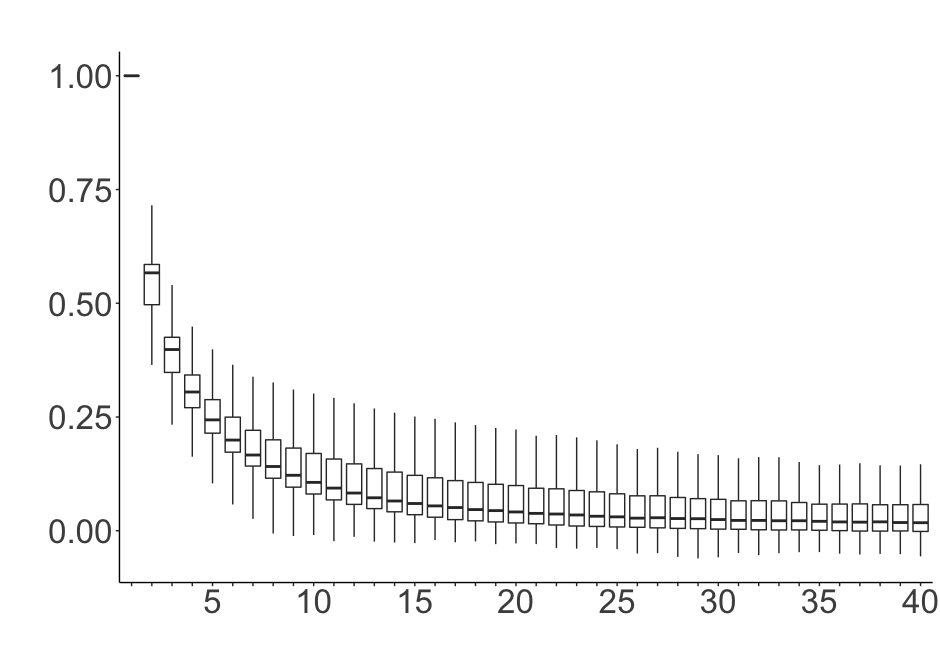}
    \put(52, 0){\scriptsize Lag}
    \put(0, 30){\rotatebox{90}{\scriptsize ACF}}
    \end{overpic}
    \caption{ACF plot of $\theta^{100}$.}
\end{subfigure}
\caption{Trace plots and autocorrelation function (ACF) plot. Each box on the ACF plot incorporates the ACF for all pixels of an image.}
\label{tp:stgp}
\end{figure}

When using our proposed Gibbs sampler for posterior estimation, we sample $\beta$ using the anti-correlation Gaussian data augmentation technique, sample $\sigma^2$ and $\tau$ from their corresponding inverse-gamma full conditional distributions, and sample each of $\xi$ and $\kappa_0$ using the slice sampling algorithm. 
We run $10,000$ iterations with the first $2,000$ as burn-ins, and it takes about 270 minutes
on a MacBook Pro with a 10-core CPU. We report the effective sample size per hour for four randomly selected locations of $\theta^{100}$ at $f=100$: $315.14$ at $s=3300$, $315.44$ at $s=3320$, $268.95$ at $s=3720$, and $227.16$ at $s=6700$. To further assess the mixing, we plot both the trace plots at two locations and an autocorrelation function (ACF) plot in Figure \ref{tp:stgp}. The trace plots show fast mixing and the ACF drops quickly.

\begin{figure}[H]
\begin{subfigure}[t]{.32\textwidth}
    \begin{overpic}[width=\textwidth]{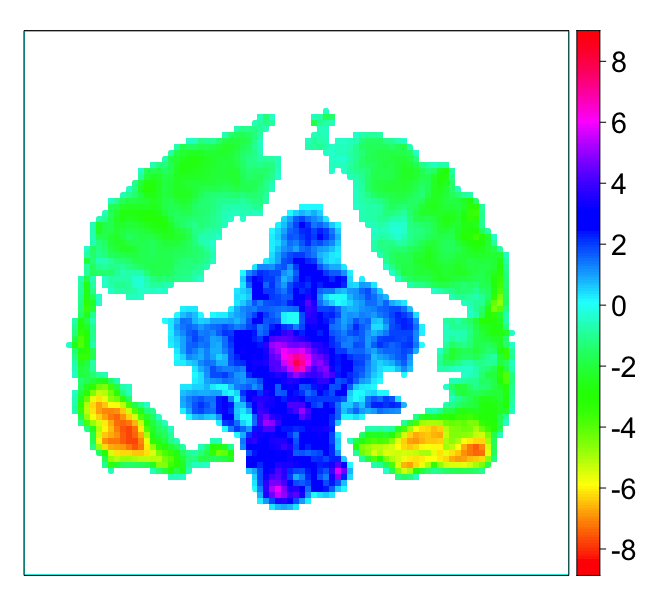}
    \end{overpic}
    \caption{Posterior estimate $\hat{\theta}^{50}$.}
\end{subfigure}
\begin{subfigure}[t]{.32\textwidth}
    \begin{overpic}[width=\textwidth]{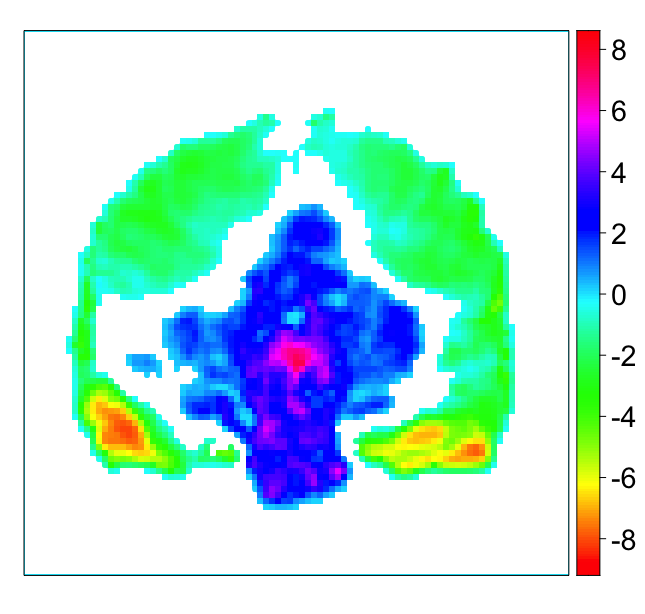}
    \end{overpic}
    \caption{Posterior estimate $\hat{\theta}^{100}$.}
\end{subfigure}
\begin{subfigure}[t]{.32\textwidth}
    \begin{overpic}[width=\textwidth]{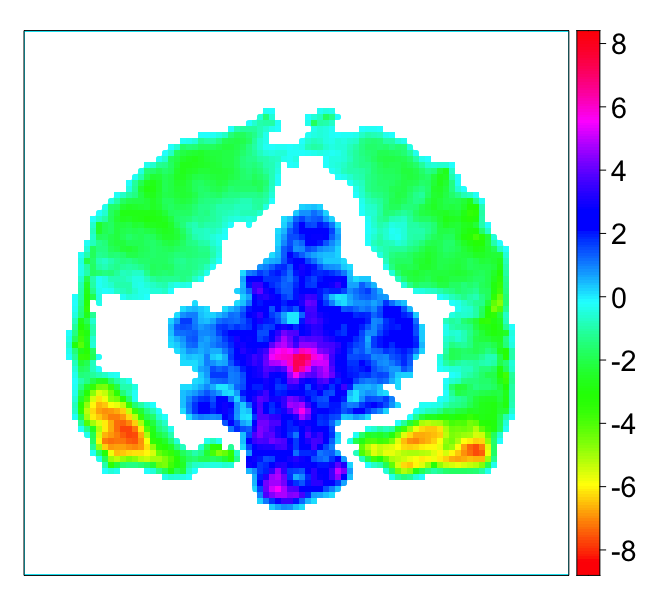}
    \end{overpic}
    \caption{Posterior estimate $\hat{\theta}^{200}$.}
\end{subfigure}
\begin{subfigure}[t]{.32\textwidth}
    \begin{overpic}[width=\textwidth]{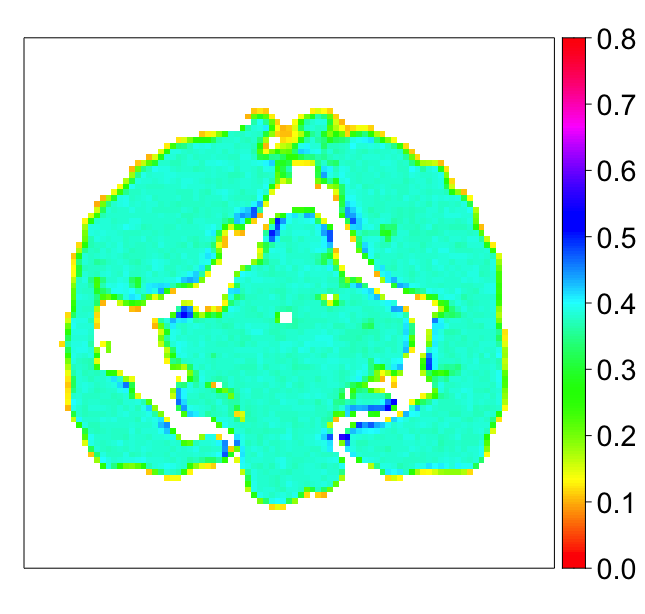}
    \end{overpic}
    \caption{Point-wise standard deviations for $\hat{\theta}^{50}$.}
\end{subfigure}
\begin{subfigure}[t]{.32\textwidth}
    \begin{overpic}[width=\textwidth]{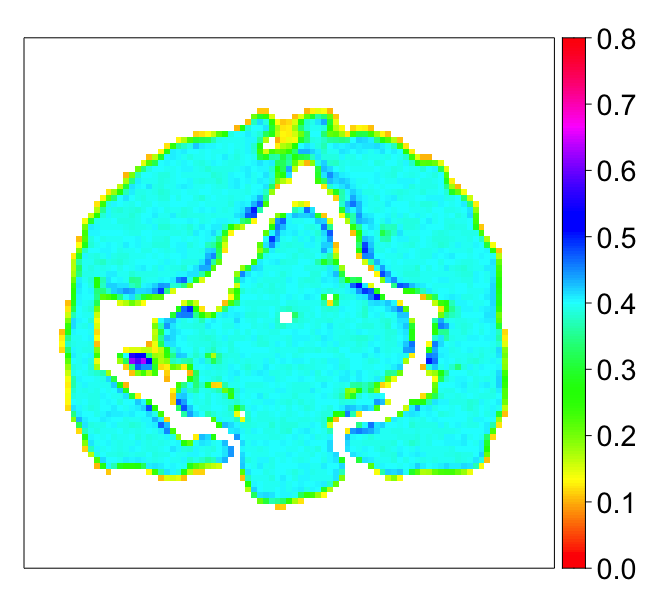}
    \end{overpic}
    \caption{Point-wise standard deviations for $\hat{\theta}^{100}$.}
\end{subfigure}
\begin{subfigure}[t]{.32\textwidth}
    \begin{overpic}[width=\textwidth]{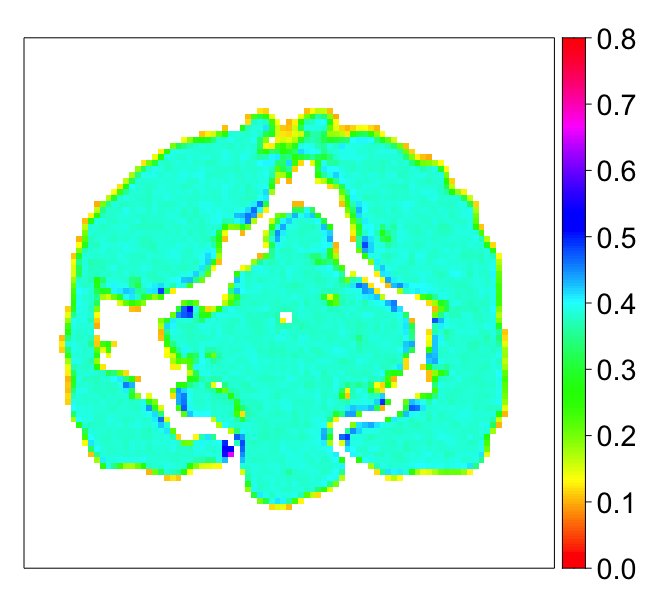}
    \end{overpic}
    \caption{Point-wise standard deviations for $\hat{\theta}^{200}$.}
\end{subfigure}

\begin{subfigure}[t]{.24\textwidth}
    \begin{overpic}[width=\textwidth]{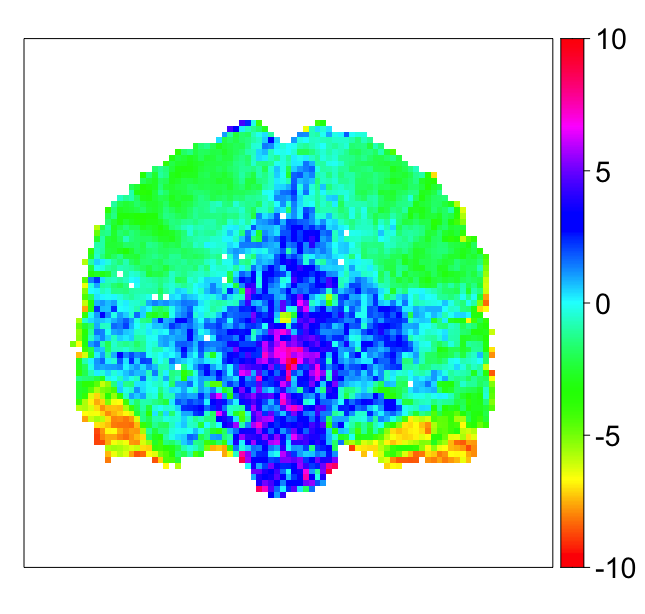}
    \end{overpic}
    \caption{Raw image $y^{100}$.}
\end{subfigure}
\begin{subfigure}[t]{.24\textwidth}
    \begin{overpic}[width=\textwidth]{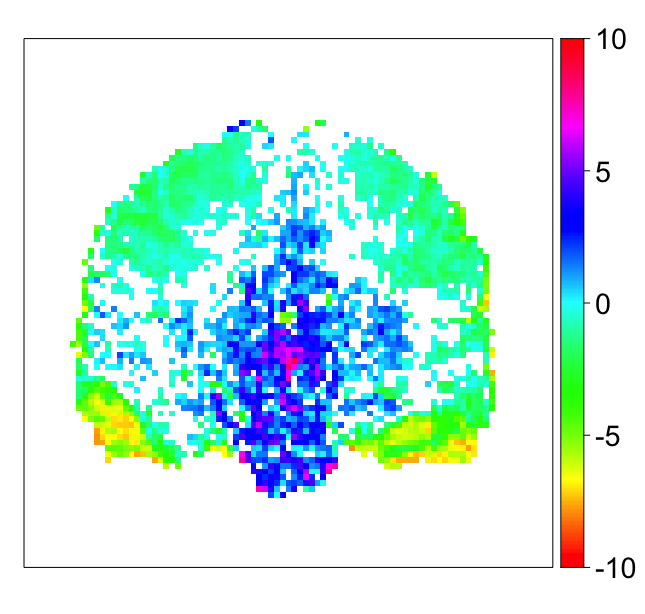}
    \end{overpic}
    \caption{Raw image $y^{100}$ after the soft-thresholding transform with threshold=1.}
\end{subfigure}
\begin{subfigure}[t]{.24\textwidth}
    \begin{overpic}[width=\textwidth]{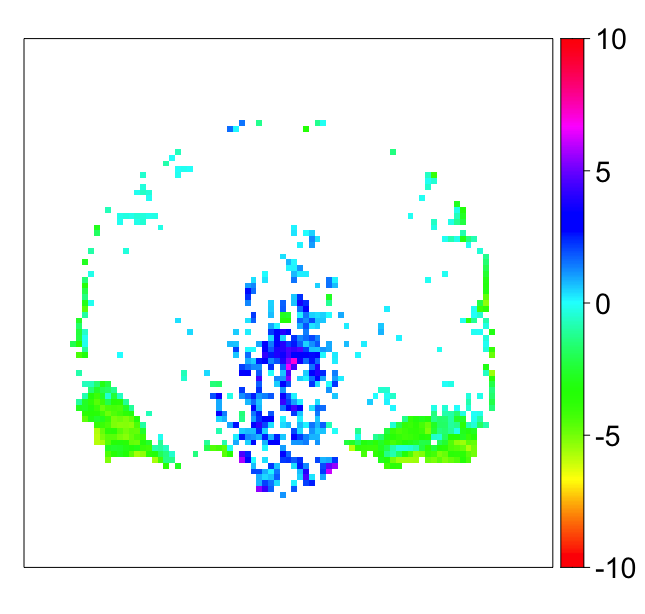}
    \end{overpic}
    \caption{Raw image $y^{100}$ after the soft-thresholding transform with threshold=3.}
\end{subfigure}
\begin{subfigure}[t]{.24\textwidth}
    \begin{overpic}[width=\textwidth]{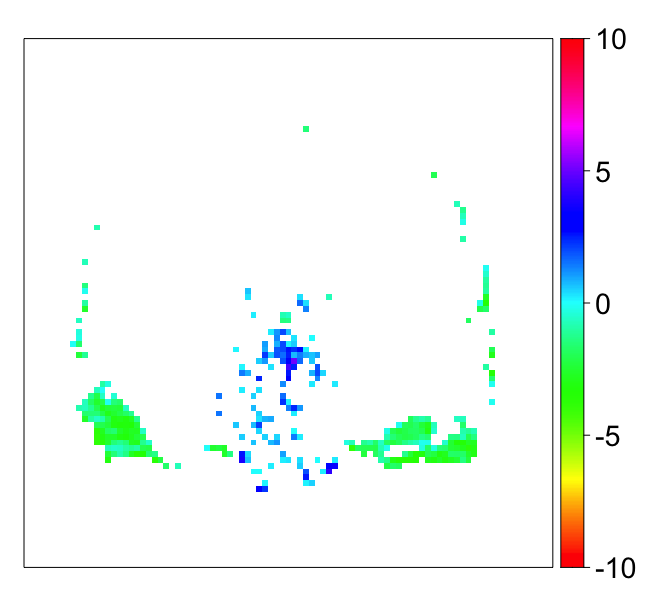}
    \end{overpic}
    \caption{Raw image $y^{100}$ after the soft-thresholding transform with threshold=5.}
\end{subfigure}
\caption{Posterior estimates, point-wise standard deviations, and direct soft-thresholding transform on a raw image. The posterior estimates by soft-thresholded GP are smooth and sparse with point-wise uncertainty estimated whereas the direct soft-thresholding transform only leads to noisy estimates.}
\label{plot:stgp}
\end{figure}

In Figure \ref{plot:stgp}, images (a)-(c) show the posterior mean estimates for three frames, and images (d)-(f) give the point-wise standard deviations. We can clearly see both smoothness and sparsity in the mean estimates, as well as how they change over time. To show the benefits of using a soft-thresholded GP, we also plot the results of simple thresholding on the raw images in (g)-(j), which lead to non-smooth estimates and scattered noisy non-zero points on the estimates.

To show how our proposed algorithm improves the computational efficiency of the soft-thresholded GP method, we also run a No-U-Turn sampler on the same model using Stan. We try NUTS with max tree depth equal to 7, it takes 68 hours to finish $10,000$ iterations, and the effective sample size per hour for estimating $\theta^{100}$ is $8.1$ on average, $9.2$ at $s=3300$, $9.3$ at $s=3320$, $7.4$ at $s=3720$, and $6.7$ at $s=6700$. Therefore,  our algorithm improves the computing speed substantially.

\section{Discussion}
The latent Gaussian form is very common in Bayesian posterior estimation. In this article, we show the advantages of using anti-correlation Gaussian to cancel out the quadratic term in the exponent.
We expect our technique to be broadly applicable beyond our presented scenarios. For example, one could potentially extend this technique to address sampling problems of Bingham distribution in Stiefel manifold  \citep{hoff2009simulation}, the costly evaluation of determinant in determinantal point process \citep{kulesza2012determinantal}, and so forth.

When it comes to linear regression variable selection using discrete spike-and-slab prior,
there is some close connection of our algorithm to the ODA algorithm
\citep{ghosh2011rao}, which uses an augmented design matrix to make the complete gram matrix diagonal. 
Since one could reparameterize a discrete spike-and-slab prior as a special case of the L1-ball prior by setting $\kappa$ according to a quantile of $\pi_0(\beta)$, we expect our geometric ergodicity result could be extended to analyzing the ODA algorithm. In comparison, our algorithm has the advantages of not having to create an augmented design matrix and easy application in non-linear models. We expect some nice properties that are further exploited in the ODA algorithm, such as the collapsed sampling step based on marginalizing out $\theta_j$ given if $b_j=1$ or $0$, could be obtained in some specific form of latent Gaussian model after we augment the anti-correlation Gaussian.

\bibliography{ref.bib}
\bibliographystyle{chicago}

\appendix

\section{Proof of Geometric Ergodicity in Linear Models}\label{geometricergodicityproof}

We wish to prove the geometric ergodicity of the Markov chain $\{r^k,t^k\}_{k=0}^\infty$. To ease the burden of notation, we denote $(r^0,t^0)$ the current state and $(r,t)$ the next state generated by first sampling $(\theta,\beta)|(r,t)$ according to (\ref{eq:b}) and (\ref{eq:theta_b}) and then sampling $(r,t)|(\theta,\beta)$ according to (\ref{eq:anti-quad-r}). Let $\|\cdot\|$ denote the Euclidean norm and define $V(x,y):=\|x+y\|^2$ for any $x,y\in\mathbb R^p$. Invoking Theorem 12 in \cite{Rosenthal1994MinorizationCA} (See also \cite{Meyn1994geometricconvergence}), it suffices to show the following two conditions hold:
\begin{enumerate}
    \item \textbf{Drift condition:} There exist $\tilde\lambda <1$ and $\tilde b<\infty$ such that $\mathbb E[V(r,t)|r^0,t^0]\leq\tilde\lambda V(r^0,t^0)+\tilde b$ for any $(r^0,t^0)\in\mathbb R^{2p}$.
    \item \textbf{Minorization condition:} For some $\tilde d>2\tilde b/(1-\tilde\lambda)$, there exist $\varepsilon>0$ and a probability measure $Q:\mathbb R^p\times\mathbb R^p\to\mathbb R$ such that $P_{r,t}((r^0,t^0),\cdot)\geq\varepsilon Q(\cdot)$ holds for all $(r^0,t^0)\in\mathbb R^{2p}$ satisfying $V(r^0,t^0)\leq\tilde d$
\end{enumerate}

\noindent\textbf{1. Drift condition:}
It is known that for a random variable $X$ following a truncated Gaussian distribution $\text{N}_{(a,b)}\left(\mu,\sigma^2\right)$, we have $\mathbb E[X]=\mu+\frac{\varphi(\alpha)-\varphi(\beta)}{\Phi(\beta)-\Phi(\alpha)}\sigma$ where $\alpha=(a-\mu)/\sigma$, $\beta=(b-\mu)/\sigma$, $\varphi(\cdot)$ is the probability density function of the standard Gaussian distribution and $\Phi (\cdot )$ is its cumulative distribution function. Besides, because the truncation reduces the variance, it follows that $\text{Var}(X)\leq\sigma^2$. Hence, we can derive an upper bound for the second moment of $X$:
\[\label{ineq:trunormal2}
\mathbb E[X^2]\leq \left(\mu+\frac{\varphi(\alpha)-\varphi(\beta)}{\Phi(\beta)-\Phi(\alpha)}\sigma\right)^2+\sigma^2.
\]

It follows that for any $j=1,\ldots,p$, 
\(
\mathbb E[\beta_j^2|b_j=1,r^0,t^0]\leq\left(\mu^0_j+\frac{\varphi(\frac{\kappa_j-\mu^0_j}{\sigma})}{1-\Phi(\frac{\kappa_j-\mu^0_j}{\sigma})}\sigma\right)^2+\sigma^2:=\zeta_j^0,
\)
where $\mu^0_j=\frac{\phi_j+r^0_{j}+\psi_j+t^0_{j}+d\kappa_j}{d+e}$ and $\sigma=\frac{1}{\sqrt{d+e}}$.

Using the asymptotic property of the Mill's ratio $\lim\limits_{x\to+\infty} \frac{\varphi(x)}{x(1-\Phi(x))}=1$ and $\lim\limits_{x\to-\infty} \frac{\varphi(x)}{x(1-\Phi(x))}=0$, we have 
\(
\lim_{\mu_j^0\to+\infty} \frac{\zeta_j^0}{(\mu_j^0)^2}=1,\text{ and }\lim_{\mu_j^0\to-\infty} \frac{\zeta_j^0}{(\mu_j^0)^2}=0.
\)

Furthermore, in view of the fact that $\lim_{|r^0_j+t^0_j|\to\infty}\frac{r^0_j+t^0_j}{d+e}/\mu_j^0=1$, we get
\(
\lim_{r_j^0+t_j^0\to+\infty}\frac{\zeta_j^0}{(\frac{r_j^0+t_j^0}{d+e})^2}=1,\text{ and }\lim_{r_j^0+t_j^0\to-\infty}\frac{\zeta_j^0}{(\frac{r_j^0+t_j^0}{d+e})^2}=0.
\)

It follows that for any real number $q>1$, there exists a real number $\tilde{C}_j>0$, such that for all $(r_j^0,t_j^0)\in\mathbb R^2$,
\(
\zeta^0_j<q\left(\frac{r_j^0+t_j^0}{d+e}\right)^2+ \tilde{C}_j. 
\)
In particular, one can take $q=\frac{1-\nu/2}{1-\nu}$ where $\nu:=1-(\lambda_p(dI-M)+\lambda_p(eI-H))^2/(d+e)^2\in(0,1)$ and $\lambda_p(A)$ denotes the largest eigenvalue of matrix $A$. Then we have

\(
\zeta^0_j<\frac{1-\nu/2}{1-\nu}\left(\frac{r_j^0+t_j^0}{d+e}\right)^2+ \tilde{C}_j. 
\)
Noting that a similar result can be reached with a different constant term than $\tilde{C}_j$ for when $b_j=-1$ is given, and that when $b_j=0$ is given, $\beta_j^2\leq\kappa_j^2$, we conclude that
\(
\mathbb E[\beta_j^2|b_j=i,r^0,t^0]\leq \frac{1-\nu/2}{1-\nu}\left(\frac{r_j^0+t_j^0}{d+e}\right)^2+C
\)
for $i\in\{-1,0,1\}$, $j\in\{1,\ldots,p\}$ and some constant $C$.
Hence,
\[\label{ineq:Ebeta2}
\begin{aligned}
\mathbb E[\|\beta\|^2|r^0,t^0]&=\sum_{j=1}^p\mathbb E[\beta_j^2|r^0,t^0]\\
&=\sum_{j=1}^p\sum_{i\in\{-1,0,1\}}\mathbb E[\beta_j^2|b_j=i,r^0,t^0]\text{Pr}(b_j=i|r^0,t^0)\\
&\leq\sum_{j=1}^p \left[\frac{1-\nu/2}{1-\nu}\left(\frac{r_j^0+t_j^0}{d+e}\right)^2+C\right]\\
&=\frac{1-\nu/2}{1-\nu}\frac{1}{(d+e)^2}\|r^0+t^0\|^2+pC.
\end{aligned}
\]

Thus, we establish the drift condition as follows:
\(
\begin{aligned}
&\mathbb E[V(r,t)|r^0,t^0]\\
=&\mathbb E[\mathbb E(V(r,t)|\theta,\beta)|r^0,t^0]\\
=&\mathbb E[\mathbb E(\|r\|^2+\|t\|^2+2r't|\theta,\beta)|r^0,t^0]\\
\stackrel{(a)}{=}&\mathbb E[\|(dI-M)\theta\|^2+\|(eI-H)\beta\|^2+2\theta'(dI-M)(eI-H)\beta|r^0,t^0]\\
&\qquad\qquad+tr(dI-M)+tr(eI-H)+tr\left(\begin{pmatrix}
O&I_p\\
I_p&O
\end{pmatrix}
\begin{pmatrix}
dI-M&O\\O&eI-H
\end{pmatrix}\right)\\
\stackrel{(b)}{\leq}&\lambda^2_p(dI-M)\mathbb E[\|\theta\|^2|r^0,t^0]+\lambda_p^2(eI-H)\mathbb E[\|\beta\|^2|r^0,t^0]\\
&\qquad\qquad +2\lambda_p(dI-M)\lambda_p(eI-H)\mathbb E[\|\theta\|\|\beta\||r^0,t^0]+C'\\
\stackrel{(c)}{\leq}&(\lambda_p(dI-M)+\lambda_p(eI-H))^2\mathbb E[\|\beta\|^2|r^0,t^0]+C'\\
\stackrel{(d)}{\leq}&(\lambda_p(dI-M)+\lambda_p(eI-H))^2\left[\frac{1-\nu/2}{1-\nu}\frac{1}{(d+e)^2}\|r^0+t^0\|^2+pC\right]+C'\\
=&\tilde{\lambda}V(r^0,t^0)+\tilde{b}
\end{aligned}
\)
where $C'=tr(dI-M+eI-H)$, $\tilde{\lambda}=\frac{1-\nu/2}{1-\nu}\frac{(\lambda_p(dI-M)+\lambda_p(eI-H))^2}{(d+e)^2}=1-\frac{\nu}{2}\in(0,1)$, and $\tilde{b}=C'+(\lambda_p(dI-M)+\lambda_p(eI-H))^2pC<\infty$. $(a)$ uses the fact that $(r,t)|(\theta,\beta)$ follows $2p$-dimensional multivariate Gaussian distribution; $(b)$ invokes the Rayleigh-Ritz theorem and for the third term, the Cauchy-Schwarz inequality; $(c)$ uses the fact that $|\theta_j|=(|\beta_j|-\kappa_j)_+\leq|\beta_j|$; $(d)$ uses (\ref{ineq:Ebeta2}).

\noindent\textbf{2. Minorization condition:}
Let $\pi(r,t|r^0,t^0)$ denote the probability density of $(r,t)$ given the last state $(r^0,t^0)$ by first drawing $(\theta,\beta)$ from $\pi(\theta,\beta|r^0,t^0)$ and then drawing $(r,t)$ from $\pi(r,t|\theta,\beta)$. For any $\tilde d>0$, define the level set $\mathcal G_{\tilde d}:=\{(r^0,t^0)\in\mathbb R^{2p}:V(r^0,t^0)\leq\tilde d\}$.

For all $(r^0,t^0)\in\mathcal G_{\tilde d}$, we have $|r^0_j+t_j^0|<\sqrt{\tilde d}$ for any $j=1,\ldots,p$.

Note that
\(
\begin{aligned}
&\pi(r,t|r^0,t^0)\\
=&\mathbb E_{(\theta,\beta|r^0,t^0)}[\pi(r,t| \theta,\beta)]\\
=&K_1\mathbb E_{(\theta,\beta|r^0,t^0)}\left[\exp\left\{-\frac{1}{2}[\theta'(dI-M)\theta+\beta'(eI-H)\beta-2r'\theta-2t'\beta]\right\}\right]\\
\stackrel{(a)}{\geq}& K_1\mathbb E_{(\theta,\beta|r^0,t^0)}\left[ \exp\left\{-\frac{1}{2}\sum_{j=1}^p[d\theta_j^2+e\beta_j^2-2r_j\theta_j-2t_j\beta_j]\right\}\right]\\
\stackrel{(b)}{=}&K_1 \prod_{j=1}^p\mathbb E_{(\theta_j,\beta_j|r^0,t^0)}\left[\exp\left\{-\frac{1}{2}[d\theta_j^2+e\beta_j^2-2r_j\theta_j-2t_j\beta_j]\right\}\right]\\
=&K_1\prod_{j=1}^p \sum_{i\in\{-1,0,1\}} \mathbb E\left[\exp\left\{-\frac{1}{2}[d\theta_j^2+e\beta_j^2-2r_j\theta_j-2t_j\beta_j]\right\}\bigg|b_j=i, r^0,t^0\right]\text{Pr}(b_j=i|r^0,t^0)\\
\geq & K_1\prod_{j=1}^p\min_{i\in\{-1,0,1\}} \underbrace{\mathbb E\left[\exp\left\{-\frac{1}{2}[d\theta_j^2+e\beta_j^2-2r_j\theta_j-2t_j\beta_j]\right\}\bigg|b_j=i, r^0,t^0\right]}_{\mathcal I_{ij}}
\end{aligned}
\)
where $K_1=(2\pi)^{-2p}|dI-M|^{\frac{1}{2}}|eI-H|^{\frac{1}{2}}\exp\{-\frac{1}{2}[r'(dI-M)^{-1}r+t'(eI-H)^{-1}t]\}$. $(a)$ follows from the fact that $M$ and $H$ are both semi-positive definite and $(b)$ uses the conditional independence over $\{(\theta_j,\beta_j)\}_{j=1}^p$ given $(r^0,t^0)$.

When $i=0$: 
\(
\begin{aligned}
\mathcal I_{0j}&=\mathbb E[\exp\{-\frac{1}{2}[e\beta_j^2-2t_j\beta_j]\}|b_j=0,r^0,t^0]\\
&\geq \min \left\{\exp\{-\frac{1}{2}[e\kappa_j^2-2t_j\kappa_j]\},\exp\{-\frac{1}{2}[e\kappa_j^2+2t_j\kappa_j]\}\right\}:=f_{0j}(r,t)
\end{aligned}
\)

When $i=1$:

\begin{align*}
\mathcal I_{1j}=&\mathbb E[\exp\{-\frac{1}{2}[(d+e)\beta_j^2-2(r_j+t_j+d\kappa_j)\beta_j+d\kappa_j^2+2r_j\kappa_j]\}|b_j=1,r^0,t^0]\\
=&\frac{\sqrt{\frac{d+e}{2\pi}}}{1-\Phi[\sqrt{d+e}(\kappa_j-\mu_j^0)]}\\
&\cdot\int_{\kappa_j}^\infty\exp\bigg\{-\frac{1}{2}[2(d+e)\beta_j^2-2(r_j+t_j+d\kappa_j+(d+e)\mu_j^0)\beta_j\\
&\qquad\qquad\qquad+d\kappa_j^2+2r_j\kappa_j+(d+e)(\mu_j^0)^2]\bigg\}d\beta_j\\
=&\frac{\sqrt{\frac{d+e}{2\pi}}}{1-\Phi[\sqrt{d+e}(\kappa_j-\mu_j^0)]}\\
&\cdot\int_{\kappa_j}^\infty \exp\left\{-\frac{2(d+e)}{2}\left(\beta_j-\frac{r_j+t_j+d\kappa_j+(d+e)\mu_j^0}{2(d+e)}\right)^2\right\}d\beta_j\\
&\cdot\exp\left\{\frac{(r_j+t_j+d\kappa_j+(d+e)\mu_j^0)^2}{4(d+e)}\right\}\exp\left\{-\frac{d\kappa_j^2+2r_j\kappa_j+(d+e)(\mu_j^0)^2}{2}\right\}\\
=&\frac{1/\sqrt{2}}{1-\Phi[\sqrt{d+e}(\kappa_j-\mu_j^0)]}\cdot\left[1-\Phi\left(\frac{\kappa_j-\frac{r_j+t_j+d\kappa_j+(d+e)\mu_j^0}{2(d+e)}}{1/\sqrt{2(d+e)}}\right)\right]\\
&\cdot\exp\left\{\frac{(r_j+t_j+d\kappa_j+(d+e)\mu_j^0)^2}{4(d+e)}\right\}\exp\left\{-\frac{d\kappa_j^2+2r_j\kappa_j+(d+e)(\mu_j^0)^2}{2}\right\}\\
\geq&\frac{1}{\sqrt{2}}\Phi\left[\frac{\frac{r_j+t_j+d\kappa_j+(d+e)\mu_j^0}{2(d+e)}-\kappa_j}{1/\sqrt{2(d+e)}}\right]\exp\left\{-\frac{d\kappa_j^2+2r_j\kappa_j+(d+e)(\mu_j^0)^2}{2}\right\}\\
\stackrel{(a)}{\geq}&\frac{1}{\sqrt{2}}\Phi\left[\frac{\frac{r_j+t_j+d\kappa_j}{2(d+e)}-\frac{1}{2}(\sqrt{\tilde d}+|\phi_j+\psi_j+d\kappa_j|)-\kappa_j}{1/\sqrt{2(d+e)}}\right]\\
&\cdot\exp\left\{-\frac{d\kappa_j^2+2r_j\kappa_j+\frac{(\sqrt{\tilde d+|\phi_j+\psi_j+d\kappa_j|})^2}{d+e}}{2}\right\}\\
:=&f_{1j}(r,t),
\end{align*}
where $(a)$ uses the fact that on $\mathcal G_{\tilde d}$, we have
\(
|\mu_j^0|\leq\frac{|r_j^0+t_j^0|+|\phi_j+\psi_j+d\kappa_j|}{d+e}\leq\frac{\sqrt{\tilde d}+|\phi_j+\psi_j+d\kappa_j|}{d+e}.
\)

When $i=-1$: Similar to the $i=1$ case, we can find a lower bound $f_{-1j}(r,t)$ for $\mathcal I_{-1j}$.

To sum up, we have proved that 
\(
\pi(r,t|r^0,t^0)\geq K_1\prod_{j=1}^p\min_{i\in\{-1,0,1\}} f_{ij}(r,t),\quad\forall (r^0,t^0)\in\mathcal G_{\tilde d},
\)
which immediately leads to the minorization condition.

\section{Additional Results}

\subsection{Alternative Algorithm for Sampling Anti-correlation Gaussian}
We now describe another efficient algorithm to sample from the anti-correlation Gaussian when $M=X'\Omega X$. Since $dI-X'\Omega X$ is positive definite, we have the following converging series:
\(
(I-  d^{-1}X'\Omega X )^{-1}  = \sum_{k=0}^{\infty}  ( d^{-1}X'\Omega X )^k.
\)

When truncating the right-hand side at $K$, we have:
\[\label{eq:truncated_cov}
\sum_{k=0}^{K}  ( d^{-1}X'\Omega X )^k
=  [ I -  ( d^{-1}X'\Omega X )^{K+1}] (I-  d^{-1}X'\Omega X )^{-1},
\]
where the relative error term $\|( d^{-1}X'\Omega X )^{K+1}\|_2\le   \rho^{K}$, with $\rho=\lambda_p(X'\Omega X)/ d$. Therefore, with a selected $D$ and a numeric tolerance $\epsilon$, we can control $\rho$ and find out a suitable truncation $\hat K= \log_\rho \epsilon$. In our code, we use $\epsilon =10^{-8}$, and $\rho=2/3$  leading to $\hat K=45$.

To generate a Gaussian with mean $0$ and covariance as in \eqref{eq:truncated_cov}, we can run the following algorithm, starting from a $p$-element $\gamma=0$.
\begin{enumerate}
  \item Simulate $\eta_1 \sim \text{N}(0,I_p)$, $\eta_2 \sim \text{N}(0,I_n),$
  \item Update $\gamma \leftarrow \eta_1 + d^{-1}X'\Omega^{1/2} \eta_2+  ( d^{-1}X'\Omega X) \gamma$.
\end{enumerate}
And it is not hard to see that each iteration increases $K$ by $2$, therefore, we only need to run the above algorithm for $23$ iterations. Afterward, we set
\(
r = (dI-X'\Omega X) d^{-1/2} \gamma + (dI-X'\Omega X)\theta,
\)
which yields a Gaussian with mean $(dI-X'\Omega X)\theta$, variance numerically very close to $(dI-X'\Omega X)$.

\subsection{Fast Convergence of the Anti-correlation Gaussian Blocked Sampler}

We investigate the convergence speed of the anti-correlation Gaussian blocked sampler. We use the linear regression simulation setting with $n=300$, $\rho=0.5$, and $p\in\{1000, 5000\}$. Figure \ref{fig:lp} shows the trace plots of log-posterior density during the burn-in stage. In both settings, the log-posterior density reaches stationarity within the first 500 iterations.

\begin{figure}[H]
\centering
\begin{subfigure}[t]{0.45\textwidth}
    \begin{overpic}[width=\textwidth]{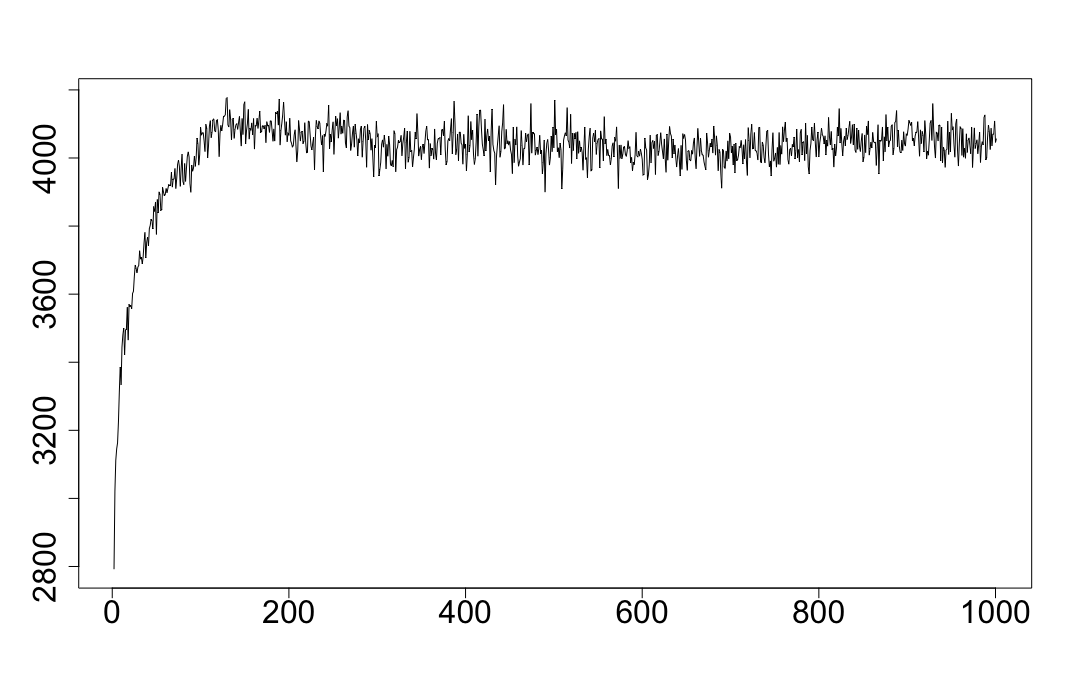}
    \put(40, -1){\scriptsize Iteration}
    \put(-2.5,20){\rotatebox{90}{\scriptsize Log posterior}}
    \end{overpic}
    \caption{Log posterior density for $p=1000$.}
\end{subfigure}
\begin{subfigure}[t]{0.45\textwidth}
    \begin{overpic}[width=\textwidth]{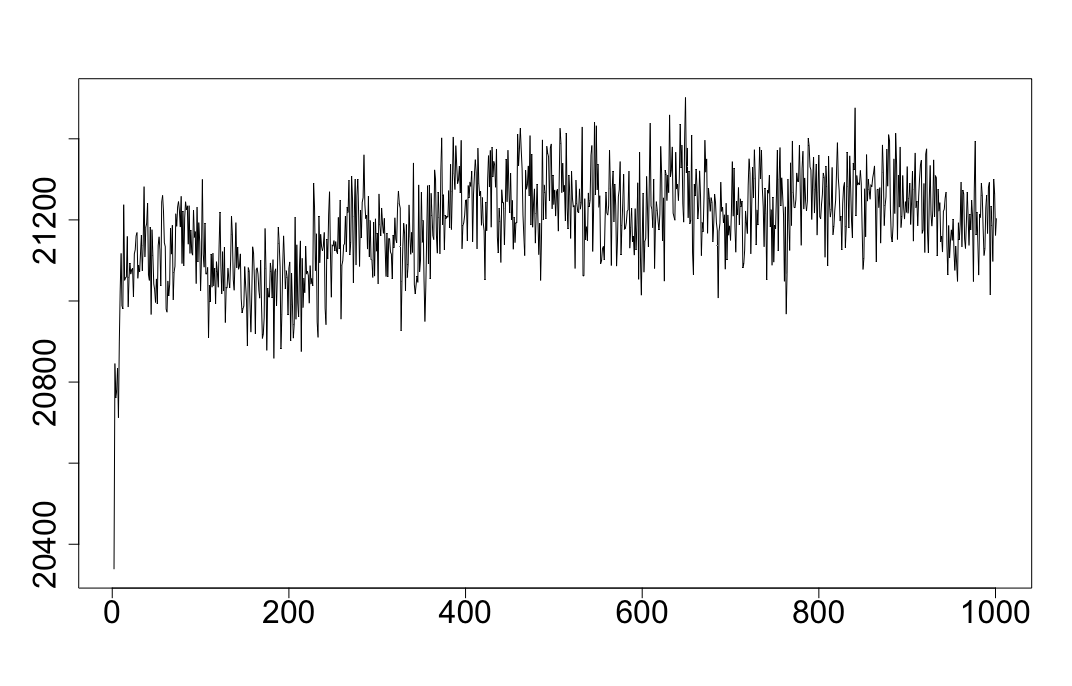}
    \put(40, -1){\scriptsize Iteration}
    \put(-2.5,20){\rotatebox{90}{\scriptsize Log posterior}}
    \end{overpic}
    \caption{Log posterior density for $p=5000$.}
\end{subfigure}
\caption{Trace plots of the log-posterior densities during the burn-in stage.}
\label{fig:lp}
\end{figure}

\subsection{Additional Simulation for Estimating Soft-thresholded Gaussian Process}
 We conduct a simulation study for comparing the performance of the anti-correlation blocked Gibbs sampler and the Metropolis-Hastings within Gibbs algorithm  (\cite{kang2018scalar}, implemented in the R package `STGP'). We generate a 30 by 30 image from a soft-thresholded Gaussian process and then add Gaussian noise to each pixel. We run both algorithms for 3,000 iterations, and treat the first 1,000 as burn-ins. The results are shown in Figure \ref{plot:stgp_simulation}. In terms of the mixing, the trace plots show that the Metropolis-Hastings within Gibbs algorithm has slightly faster mixing --- {although their implementation involves some C++ code, likely due to the high cost of Cholesky decomposition}. Both methods yield similar image estimates that are smooth and sparse, with MSE=0.02 for the anti-correlation Gaussian and MSE=0.04 for the MH-within-Gibbs sampler.

\begin{figure}[H]
\begin{subfigure}[t]{.48\textwidth}
    \begin{overpic}[width=\textwidth]{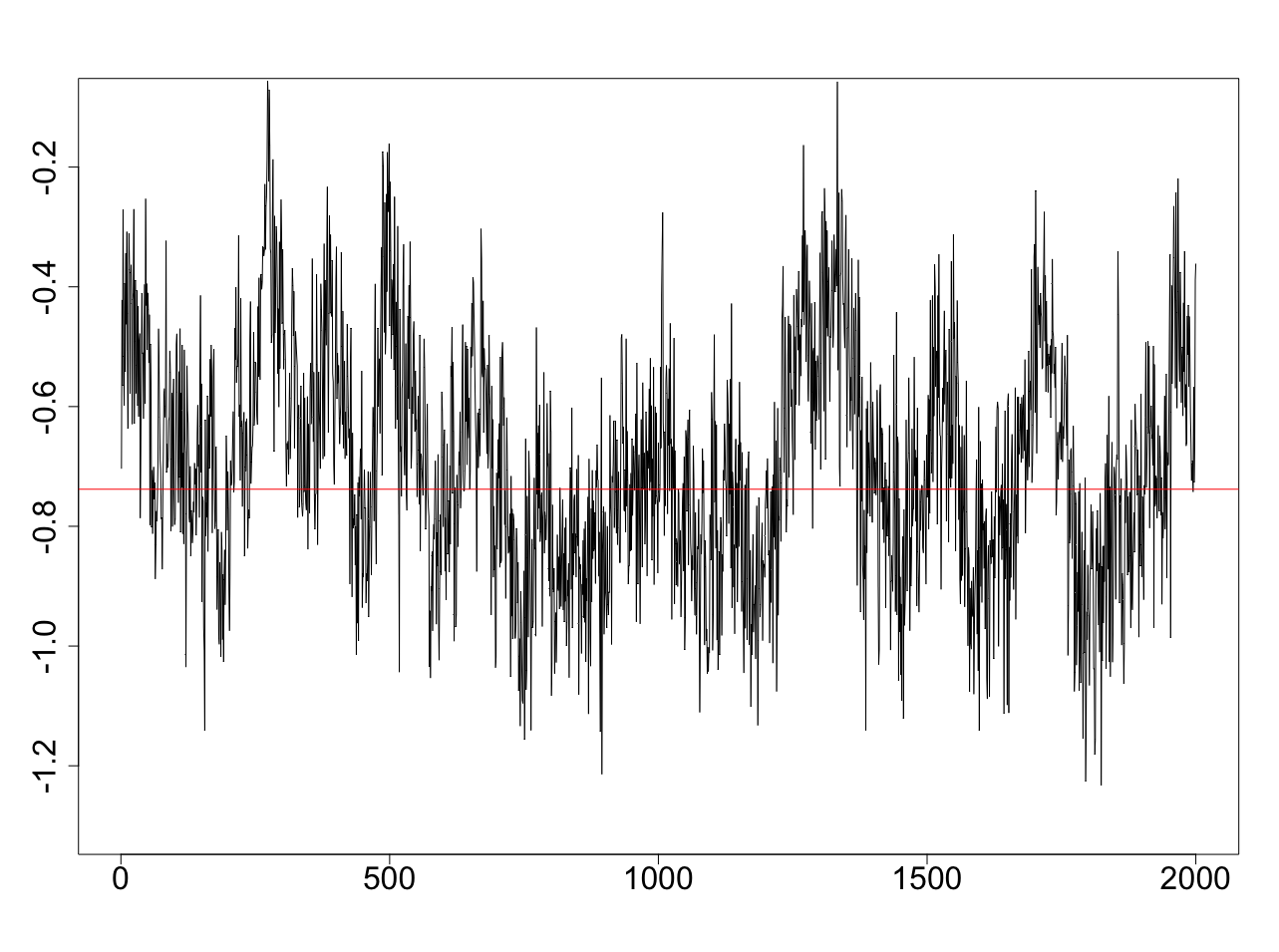}
    \put(40, -1){\scriptsize Iteration}
    \end{overpic}
    \caption{Trace plot of $\theta_{116}$ for the anti-correlation Gaussian. The horizontal line denotes the ground truth.}
\end{subfigure}
\begin{subfigure}[t]{.48\textwidth}
    \begin{overpic}[width=\textwidth]{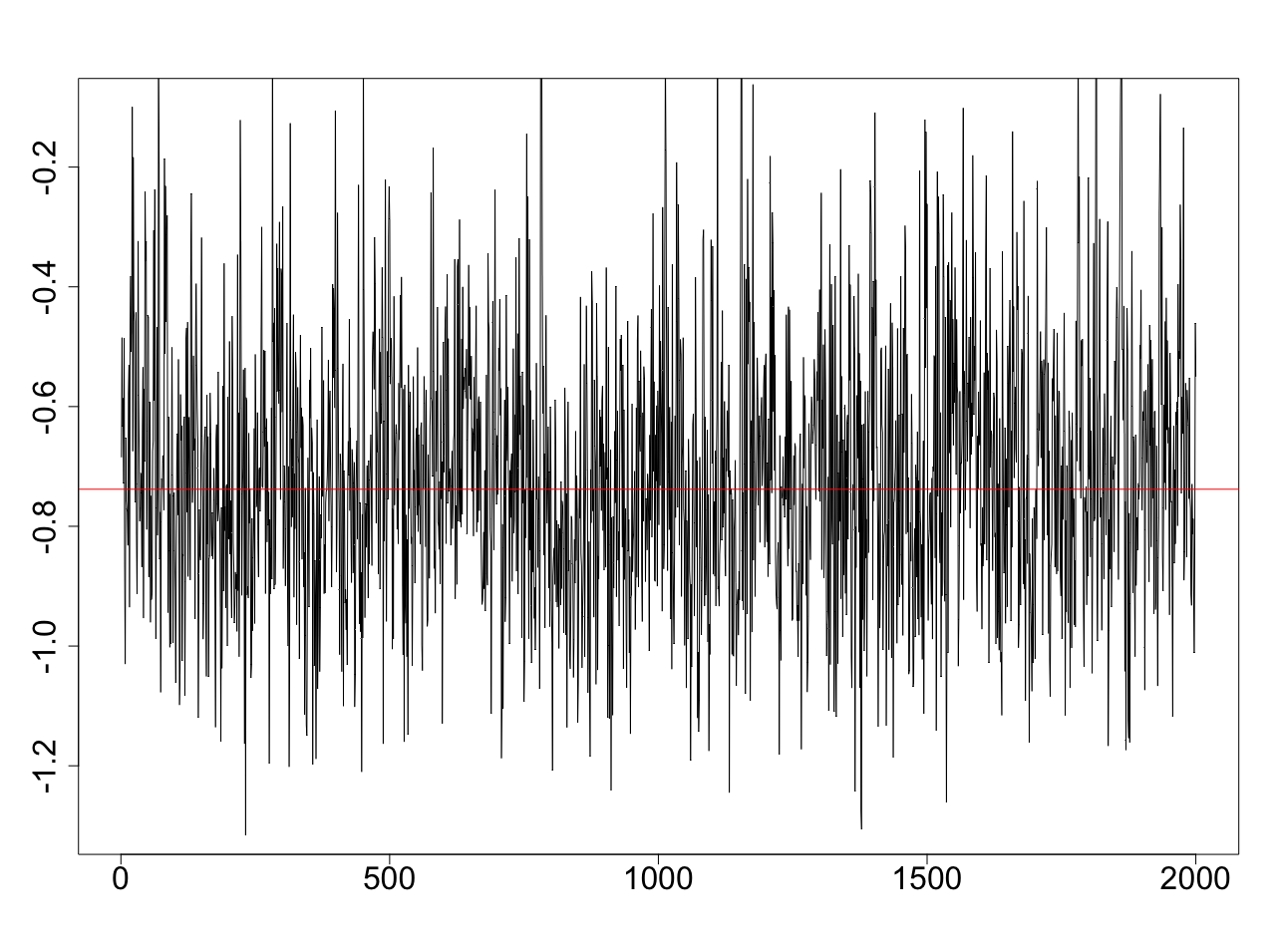}
    \put(40, -1){\scriptsize Iteration}
    \end{overpic}
    \caption{Trace plot of $\theta_{116}$ for the MH within Gibbs sampler. The horizontal line denotes the ground truth.}
\end{subfigure}
\begin{subfigure}[t]{.24\textwidth}
    \begin{overpic}[width=\textwidth]{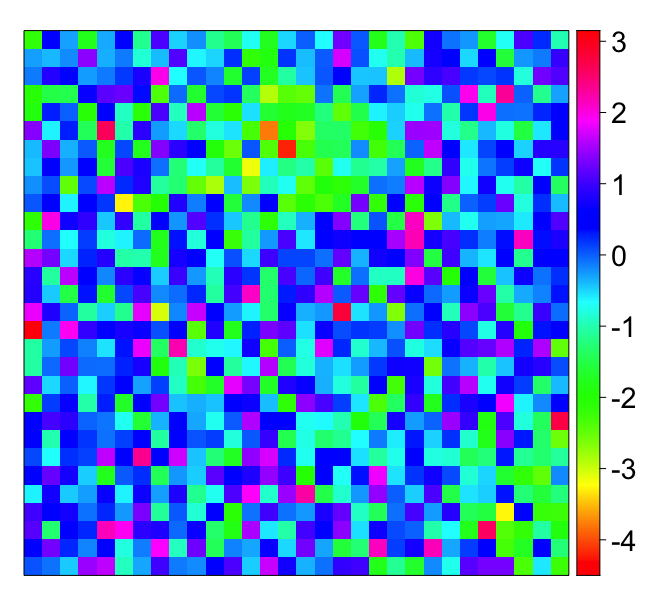}
    \end{overpic}
    \caption{Simulated data $y$.}
\end{subfigure}
\begin{subfigure}[t]{.24\textwidth}
    \begin{overpic}[width=\textwidth]{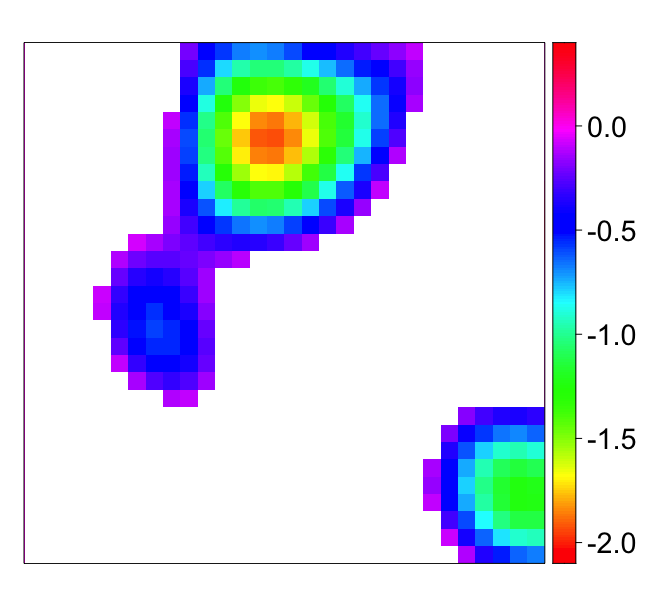}
    \end{overpic}
    \caption{Ground truth $\theta$.}
\end{subfigure}
\begin{subfigure}[t]{.24\textwidth}
    \begin{overpic}[width=\textwidth]{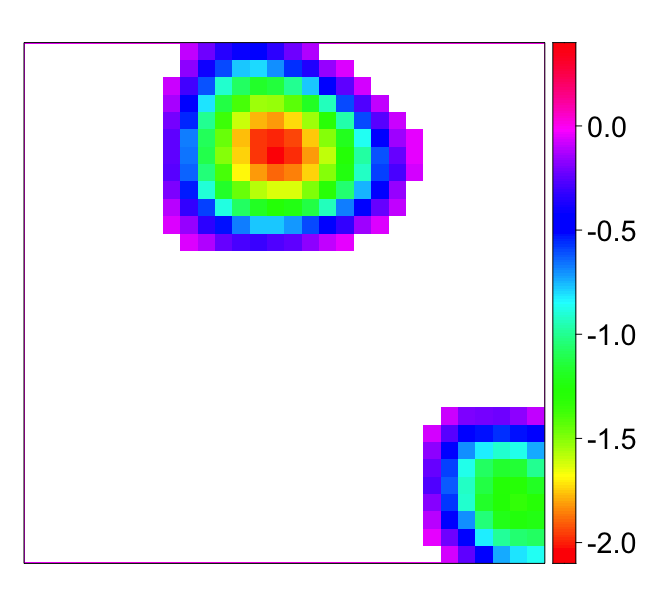}
    \end{overpic}
    \caption{Posterior estimate $\hat{\theta}$ produced by the anti-correlation Gaussian.}
\end{subfigure}
\begin{subfigure}[t]{.24\textwidth}
    \begin{overpic}[width=\textwidth]{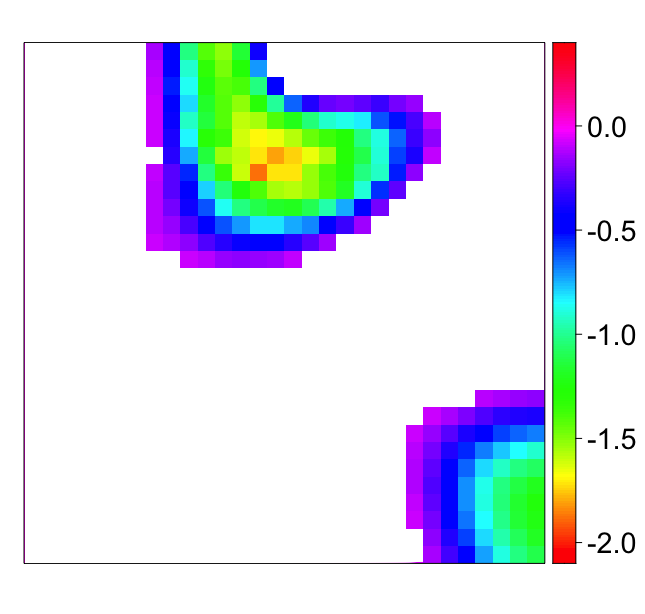}
    \end{overpic}
    \caption{Posterior estimate $\hat{\theta}$ produced by the MH-within-Gibbs sampler.}
\end{subfigure}
\begin{subfigure}[t]{.24\textwidth}
    \begin{overpic}[width=\textwidth]{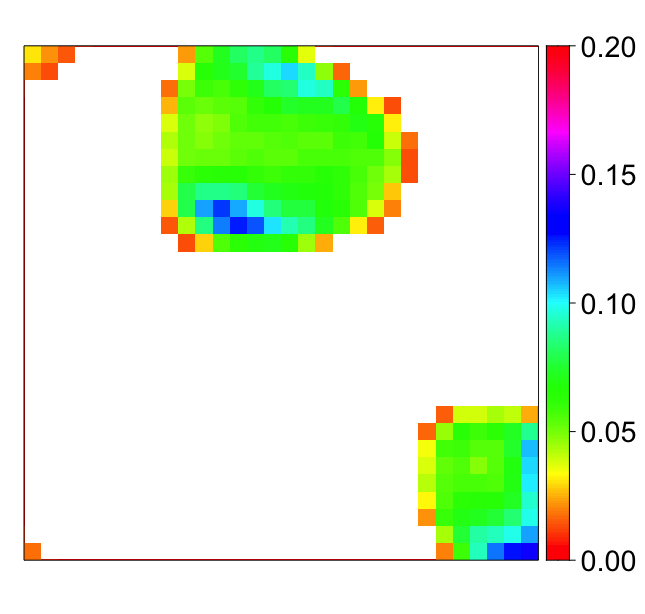}
    \end{overpic}
    \caption{Point-wise posterior variances for $\hat{\theta}$ produced by the anti-correlation Gaussian.}
\end{subfigure}
\begin{subfigure}[t]{.24\textwidth}
    \begin{overpic}[width=\textwidth]{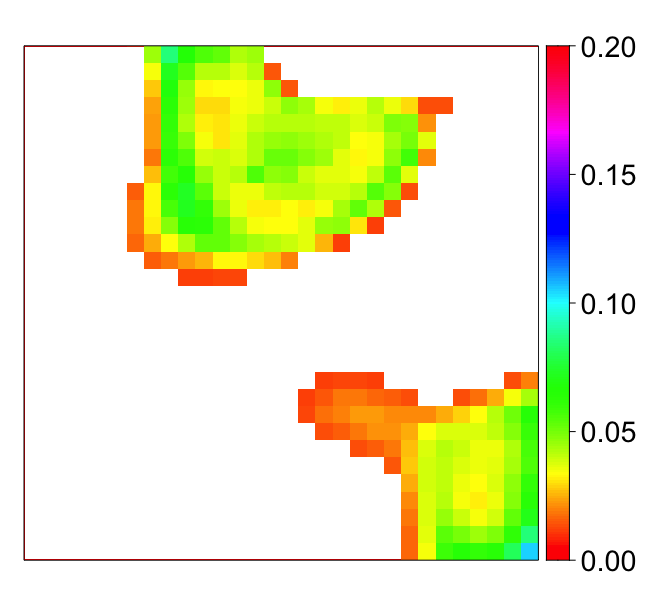}
    \end{overpic}
    \caption{Point-wise posterior variances for $\hat{\theta}$ produced by the MH-within-Gibbs sampler.}
\end{subfigure}
\begin{subfigure}[t]{.24\textwidth}
    \begin{overpic}[width=\textwidth]{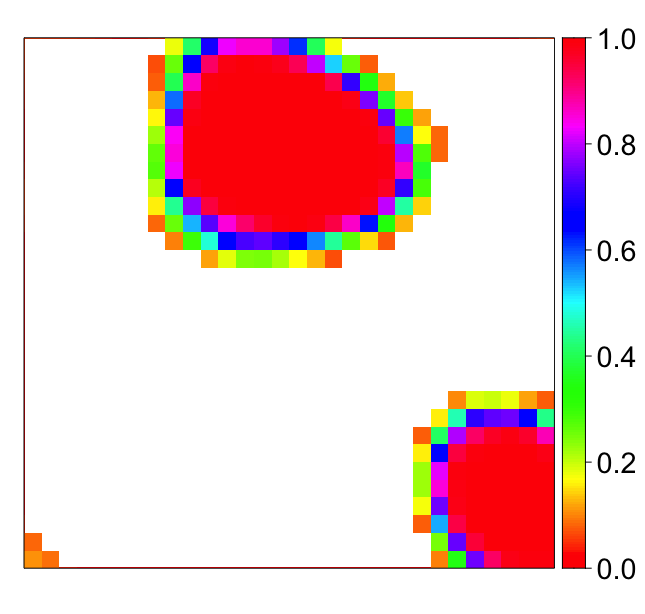}
    \end{overpic}
    \caption{$P(\theta_j\not=0|y)$ produced  by the anti-correlation Gaussian.}
\end{subfigure}
\begin{subfigure}[t]{.24\textwidth}
    \begin{overpic}[width=\textwidth]{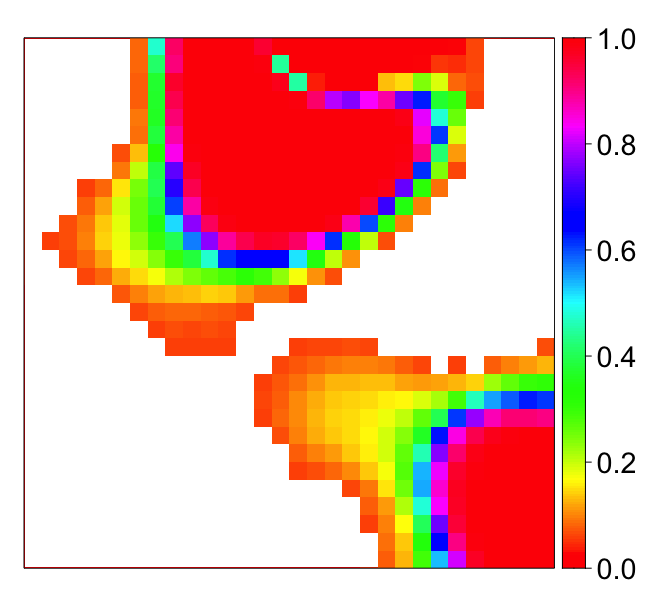}
    \end{overpic}
    \caption{$P(\theta_j\not=0|y)$ produced by the MH-within-Gibbs sampler.}
\end{subfigure}
\caption{Trace plots, data, ground truth, posterior estimates, point-wise variances, and probabilities that an entry of $\theta$ is non-zero.}
\label{plot:stgp_simulation}
\end{figure}

\subsection{Benchmark: Sampling Truncated Multivariate Gaussian using Anti-correlation Gaussian}\label{sec:truncMVN}
A comparison regarding the computational efficiency is conducted between the anti-correlation Gaussian and the `rtmvnorm()' function in the R package `tmvtnorm' (See \cite{Wilhelm2010tmvtnormAP} for details). The argument ``algorithm " in the function is set to be `gibbsR' for a fair comparison in R language; otherwise, it uses Fortran to accelerate the function. We sample from $p$-dimensional truncated multivariate Gaussian distribution with zero mean, identity covariance matrix, and truncation area $R=\prod_{i=1}^p(-4,-3]$. For such specifications, rejection sampling breaks down even when $p$ is as small as 10. We conducted $20$ simulations with $p=10$ for each algorithm. For each simulation, we do 12,000 iterations with 2,000 burn-ins and record the total clock wall time (in seconds). Figure \ref{plot:truncMVN} presents the trace plots of one randomly selected element of the parameter and compares the effective sample size per computing time (ESS/s) for anti-correlation Gaussian and rtmvnorm. Both algorithms enjoy fast mixing; however, anti-correlation Gaussian has an average ESS/s more than twice the average ESS/s for rtmvnorm across the 20 replications.

\begin{figure}[H]
\begin{subfigure}[t]{0.32\textwidth}
    \begin{overpic}[width=\textwidth]{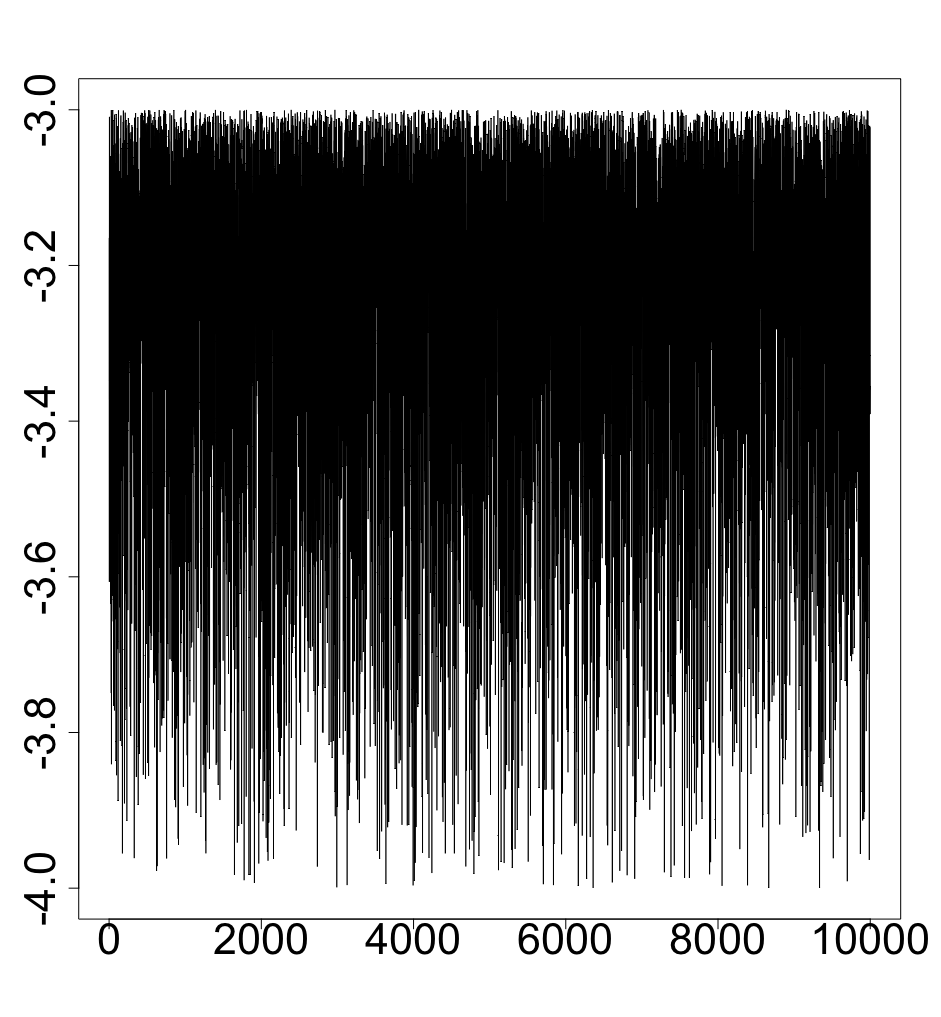}
    \put(40, -1){\scriptsize Iteration}
    \end{overpic}
    \caption{Traceplot of the second element of the parameter using Anti-Corr Gaussian.}
\end{subfigure}
\begin{subfigure}[t]{0.32\textwidth}
    \begin{overpic}[width=\textwidth]{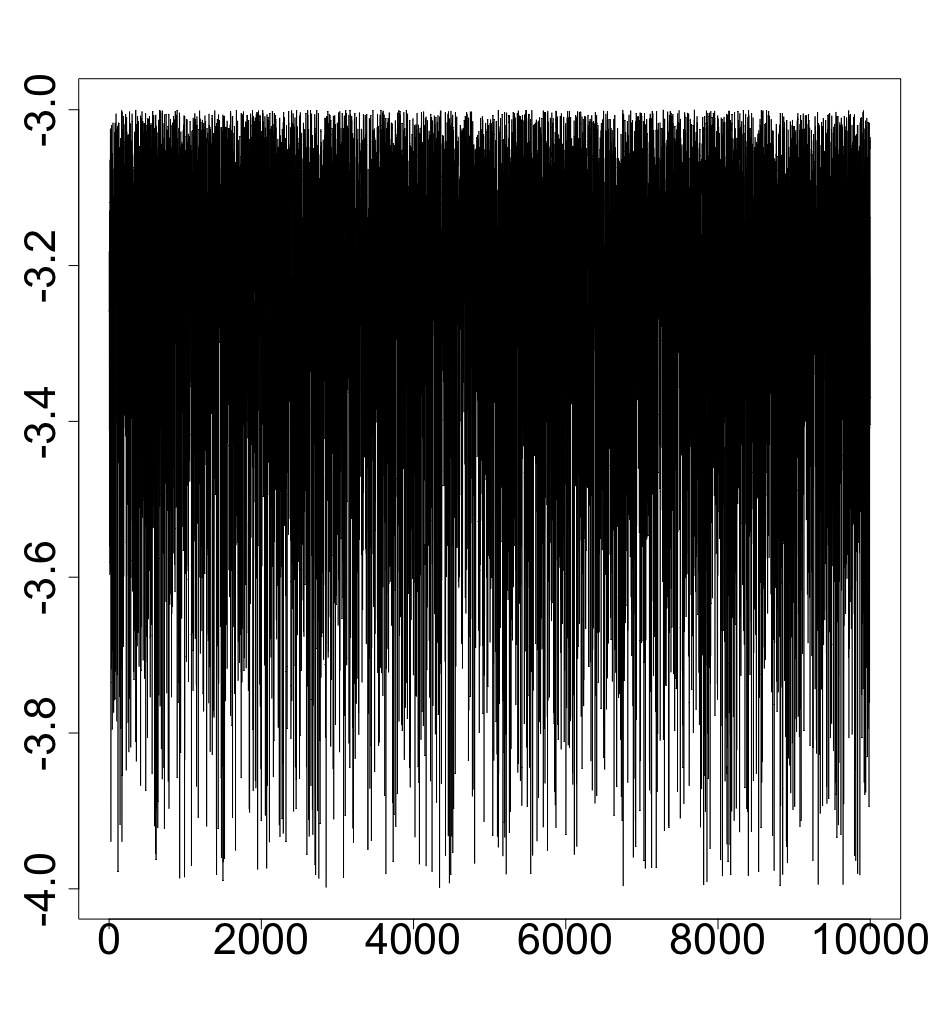}
    \put(40, -1){\scriptsize Iteration}
    \end{overpic}
    \caption{Traceplot of the second element of the parameter using rtmvnorm.}
\end{subfigure}
\begin{subfigure}[t]{0.31\textwidth}
    \begin{overpic}[width=\textwidth]{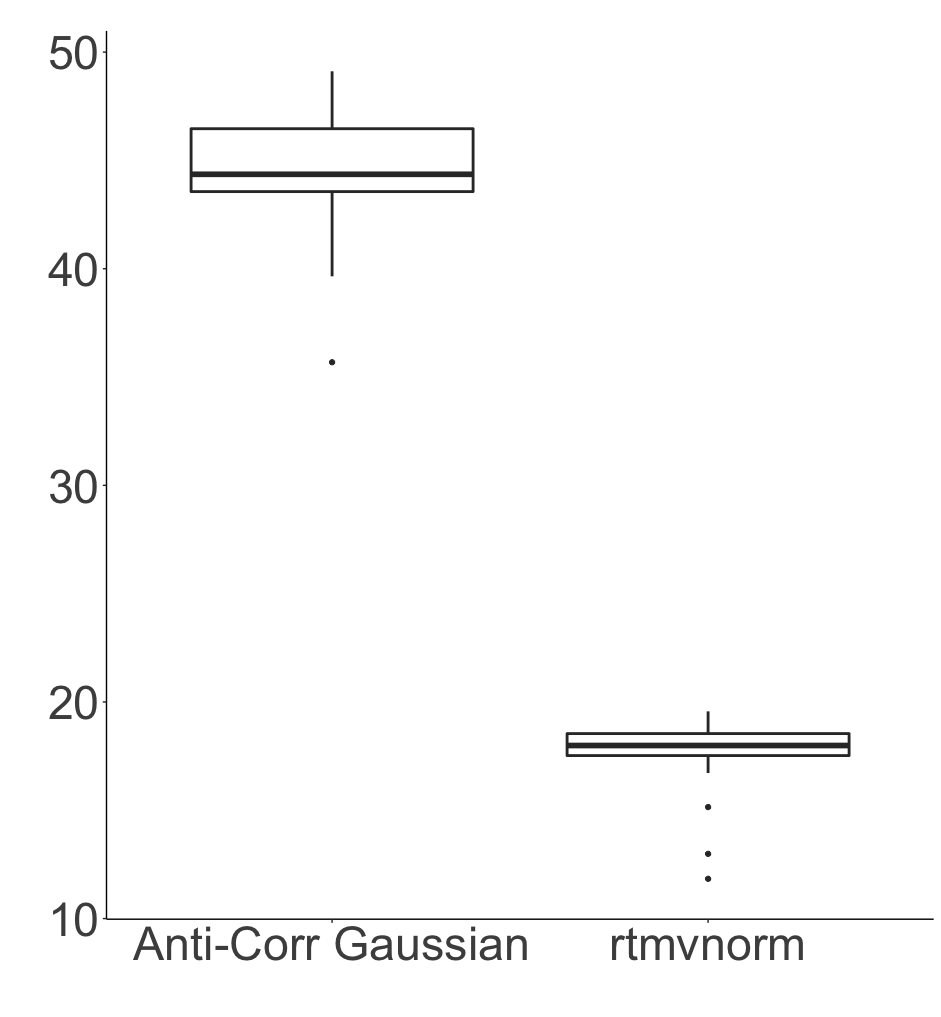}
    \put(40, 0){\scriptsize Algorithm}
    \put(0, 37){\rotatebox{90}{\scriptsize ESS/s ($\times$1,000)}}
    \end{overpic}
    \caption{Average ESS/s.}
\end{subfigure}
\caption{Traceplots and effective sample size per computing time (ESS/s) for anti-correlation Gaussian Gibbs sampler and rtmvnorm.}
\label{plot:truncMVN}
\end{figure}

\subsection{Stochastic Search Variable Selection (SSVS) Algorithm in Linear Regression}
For the original spike-and-slab priors, the stochastic search variable selection algorithm \citep{george1995stochastic} (a collapsed Gibbs sampler) enjoys excellent mixing and fast computing speed. On the other hand, the SSVS algorithm is limited to linear models that assume iid latent binary events $b_j=1(\theta_j\neq 0)$ for $j=1,\ldots,p$, whereas our algorithm can be applied to general settings with dependent $b_j$'s. For example, in a soft-thresholded Gaussian process prior, there is a strong dependency between $\theta_{s}=0$ and $\theta_{s'}=0$ if the two spatial locations  $s$ and $s'$ are close to each other. We provide a comparison between SSVS and anti-correlation Gaussian under the linear regression setting as in Section \ref{sec:lin_reg} with $p=500$, $c=3$, and $\rho=0.5$. For SSVS, we use the `lm.spike()' function in the R package `BoomSpikeSlab' \citep{scott2014predicting}. Figure \ref{plot:SSVS} shows that both anti-correlation Gaussian and SSVS work well in terms of the estimation accuracy of $\theta$, while SSVS has a faster mixing.

\begin{figure}[H]
\centering
\begin{subfigure}[t]{0.45\textwidth}
    \begin{overpic}[width=\textwidth]{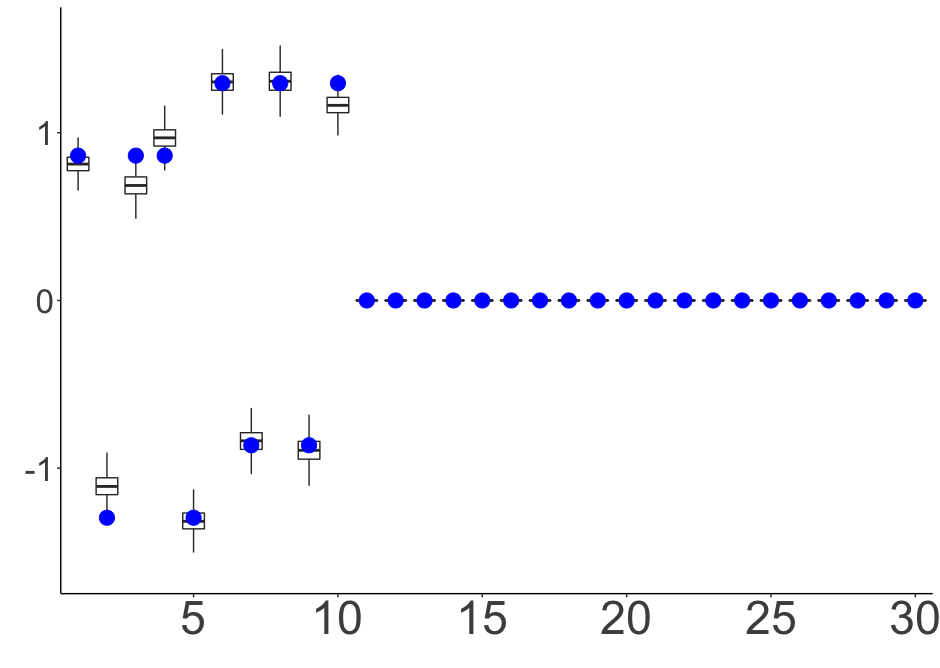}
    \put(50, -1){\scriptsize Index $j$}
    \put(-1, 37){{$\hat{\theta}_j$}}
    \end{overpic}
    \caption{Estimates $\hat{\theta}_j$ ($1\leq j\leq 30$) using anti-correlation Gaussian.}
\end{subfigure}
\begin{subfigure}[t]{0.45\textwidth}
    \begin{overpic}[width=\textwidth]{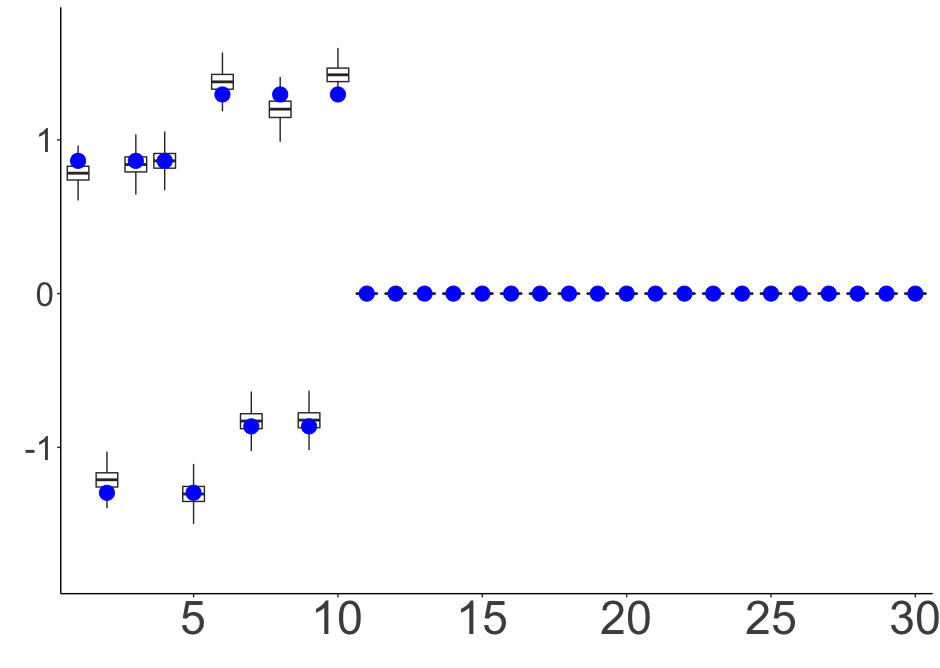}
    \put(50, -1){\scriptsize Index $j$}
    \put(-1, 37){{$\hat{\theta}_j$}}
    \end{overpic}
    \caption{Estimates $\hat{\theta}_j$ ($1\leq j\leq 30$) using SSVS.}
\end{subfigure}
\begin{subfigure}[t]{0.45\textwidth}
    \begin{overpic}[width=\textwidth]{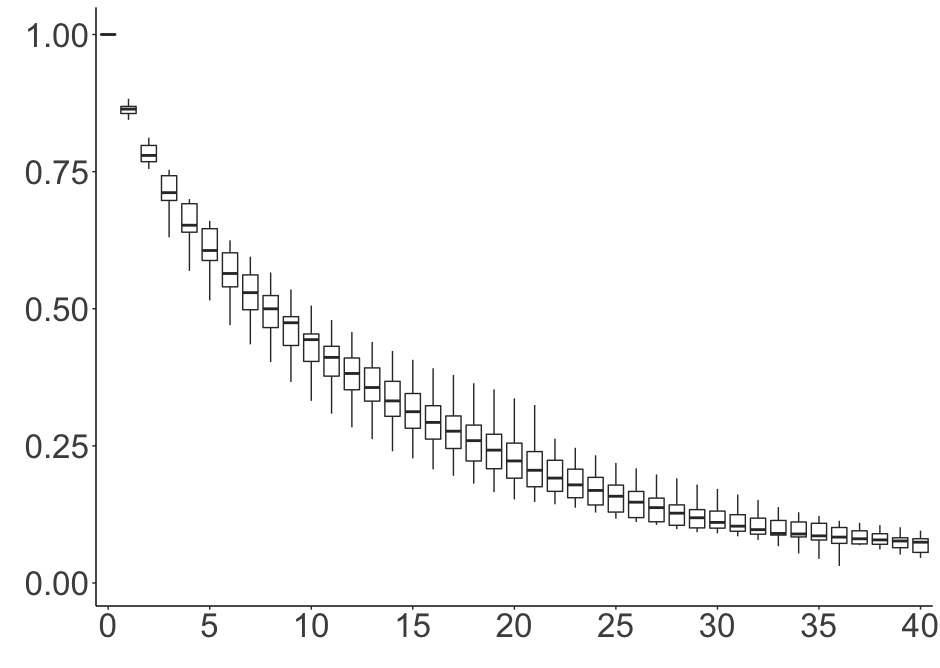}
    \put(40, 0){\scriptsize Lag}
    \end{overpic}
    \caption{ACF plot of $\theta_{1:10}$, the non-zero part of $\theta$, using anti-correlation Gaussian.}
\end{subfigure}
\begin{subfigure}[t]{.45\textwidth}
    \begin{overpic}[width=\textwidth]{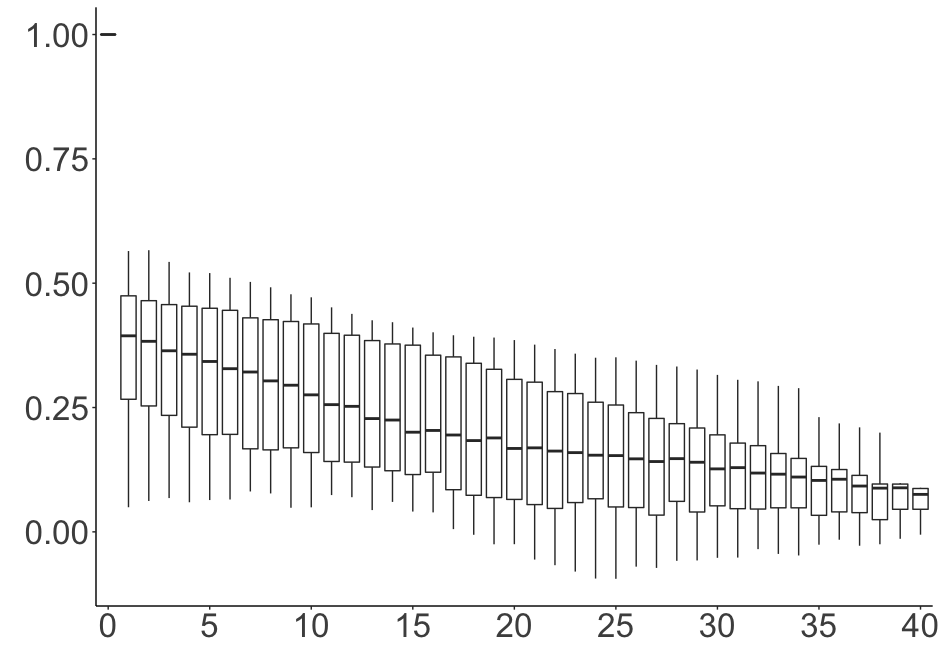}
    \put(40, -1){\scriptsize Lag}
    \end{overpic}
    \caption{ACF plot of $\theta_{1:10}$, the non-zero part of $\theta$, using SSVS.}
\end{subfigure}
\caption{Estimates and ACF plots for anti-correlation Gaussian and SSVS. The blue points denote the ground truth. Each box on the plots of estimates incorporates all sample values after the burn-in stage.}
\label{plot:SSVS}
\end{figure}

\subsection{Mixing under Different Values of Tuning Parameters $d$ and $e$}\label{mixing_de}

We empirically assess the effects of the values of $d$ (or $e$) on the mixing. We use the linear regression simulation setting with $p=50, \rho=0.5$, and $n=300$ for experiments. Since $d\sigma^2$ should be larger than $\lambda_p(X^TX)$, we pick $d=\lambda_p(X^TX)/\sigma^2+\varepsilon$. When $\varepsilon$ takes value from $\{10^{-6},10^{-4},10^{-2},1,10\}$, we see almost no difference in terms of mixing performance. 
When $\varepsilon= 100$, the algorithm fails to converge to the ground truth, likely due to numerical overflow.

\begin{figure}[H]
\centering
\includegraphics[width=0.8\linewidth]{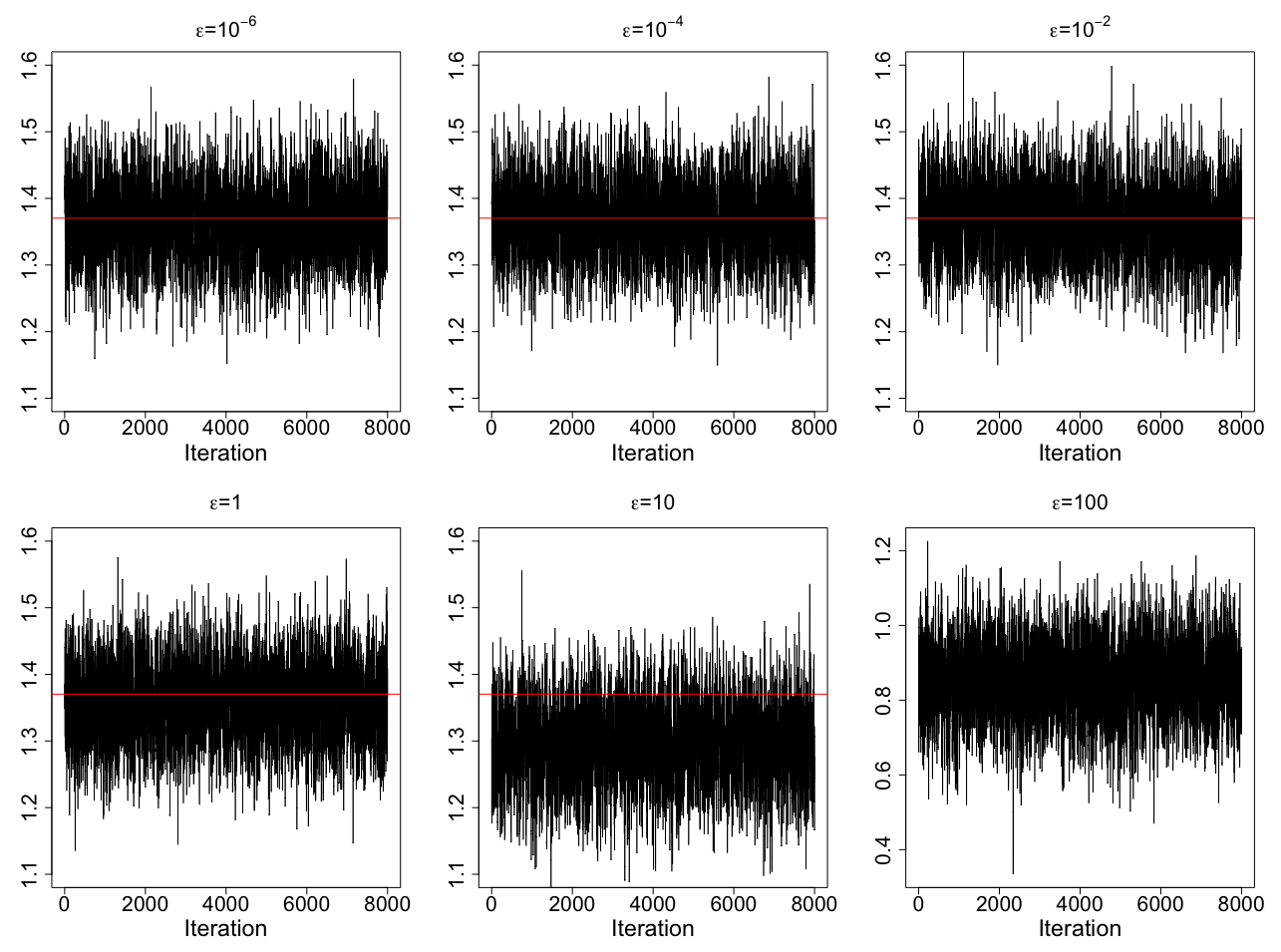}
\caption{Under $d=\lambda_p(X^TX)/\sigma^2+\varepsilon$ with $\varepsilon\in \{10^{-6},10^{-4},10^{-2},1,10,10\}$, the mixings of anti-correlation blocked Gibbs sampler have almost no difference. }
\label{plot:d}
\end{figure}

\end{document}